\documentclass[10pt]{article}

\usepackage{graphicx}
\usepackage[T1]{fontenc}
\usepackage[utf8]{inputenc}
\usepackage{amssymb}
\usepackage{amsfonts}
\usepackage{dsfont}
\usepackage{mathtools}
\usepackage{amsthm}
\usepackage{amsmath}
\usepackage{yhmath}
\usepackage{relsize}
\usepackage{textcomp}
\usepackage{eurosym}
\usepackage{stmaryrd}
\usepackage{xcolor}
\usepackage[multiple]{footmisc}

\usepackage{bigints}
\usepackage{geometry}
\usepackage[font=small,labelfont=bf]{caption}

\newtheorem{lemma}{Lemma}
\newtheorem{cor}{Corollary}

\newtheorem{prop}{Proposition}

\newtheorem{rem}{Remark}

\begin{document}

\title{Automated Market Making: the case of Pegged Assets}

\author{Philippe \textsc{Bergault}\footnote{Université Paris Dauphine-PSL, Ceremade, 75116, Paris, France, bergault@ceremade.dauphine.fr.} \and Louis \textsc{Bertucci}\footnote{Institut Louis Bachelier, 75002, Paris, France, louis.bertucci@institutlouisbachelier.org.} \and David \textsc{Bouba}\footnote{Swaap Labs, d@swaap.finance.} \and Olivier \textsc{Guéant}\footnote{Université Paris 1 Panthéon-Sorbonne, UFR 27 Mathématiques et Informatique, Centre d'Economie de la Sorbonne, Paris, France, olivier.gueant@univ-paris1.fr.} \and Julien \textsc{Guilbert}\footnote{Swaap Labs, julien@swaap.finance.}}
\date{}

\maketitle
\setlength\parindent{0pt}

\begin{abstract}

In this paper, we introduce a novel framework to model the exchange rate dynamics between two intrinsically linked cryptoassets, such as stablecoins pegged to the same fiat currency or a liquid staking token and its associated native token. Our approach employs multi-level nested Ornstein-Uhlenbeck (OU) processes, for which we derive key properties and develop calibration and filtering techniques. Then, we design an automated market maker (AMM) model specifically tailored for the swapping of closely related cryptoassets. Distinct from existing models, our AMM leverages the unique exchange rate dynamics provided by the multi-level nested OU processes, enabling more precise risk management and enhanced liquidity provision. We validate the model through numerical simulations using real-world data for the USDC/USDT and wstETH/WETH pairs, demonstrating that it consistently yields efficient quotes. This approach offers significant potential to improve liquidity in markets for pegged assets.\\

\vspace{3mm}
\vspace{8mm}

\end{abstract}

\section{Introduction}

Over the past decade, digital finance and decentralized finance have surged, leading to what can be described, paraphrasing~\cite{maurer2017blockchains}, as a ``Cambrian explosion'' of blockchain platforms, centralized cryptocurrency exchanges, decentralized applications, and, naturally, cryptocurrencies and other cryptoassets (hereafter referred to as cryptos). Today, there is a wide variety of blockchain platforms, each distinguished by its unique design, consensus mechanism, speed and security features, and the extent to which third-party development is possible. Numerous centralized crypto exchanges now compete to attract an ever-growing number of crypto investors, while a multitude of decentralized applications offer tools and services for trading, staking, lending, borrowing, and more. The number of cryptos in circulation is remarkable, with estimates ranging between 10,000 and 20,000 cryptos and a combined market capitalization of approximately 2~T\$ as of September 2024.\footnote{Source: coingecko.com as of September 2024.}\\

The vast majority of cryptos, even more so than traditional financial assets such as stocks or foreign currencies, are subject to significant trends and pronounced volatility. Their dynamics are influenced not only by international economic and geopolitical conditions but also by factors specific to the crypto sector, including its innovations, promises, and scandals. These influences are often further magnified by speculation and what Keynes famously referred to as ``animal spirits.''\\

Although many cryptos exhibit high to extreme volatility against major fiat currencies -- often characterized by sharp rises followed by sudden drops or even inevitable collapses -- not all cryptos behave in this manner. By design, some cryptos are intended to be far more stable and are known as stablecoins. Stablecoins are designed to maintain, at least in principle, a stable value relative to a specific asset or group of assets. Most aim to be a digital equivalent of fiat currencies on various blockchains, with the US dollar being the most prominent example.\footnote{There are also stablecoins associated with other currencies, such as the euro (\euro), though their total market capitalization does not exceed a few hundred million euros. Stablecoins linked to a basket of currencies are also discussed in the literature but are not expected to gain significant traction, according to~\cite{baughman2023global}.} The leading stablecoins by market capitalization are, as of September 2024,\footnote{Source: coingecko.com as of September 2024.} Tether (USDT) at 118 B\$, USD Coin (USDC) at 35 B\$, and Dai (DAI) at 5 B\$, together accounting for nearly 160 B\$. This is followed by a long tail of smaller stablecoins, bringing the total market capitalization to approximately 170 B\$. Given that stablecoins had a total market capitalization of only 3 B\$ in 2019, their increasingly important role in trying to provide a safe haven for crypto investors and bridging traditional finance and digital finance is undeniable.\\

Even when intended to represent the same fiat currency, stablecoins can vary in several important aspects. In terms of collateral, stablecoins can be fiat-collateralized, meaning fully backed by fiat currencies (though in practice, the exact composition of the collateral and exposure to market, credit, liquidity, and operational risks may be more or less transparent), crypto-collateralized, i.e., backed by other cryptos, or undercollateralized, as seen with algorithmic stablecoins, which aim to maintain their peg through arbitrage incentives for users. Collateral management may be either centralized or decentralized (see~\cite{kozhan2021decentralized} for a study of MakerDAO's DAI token, the most prominent decentralized stablecoin). Creation and redemption can be restricted to a list of authorized arbitrageurs (see~\cite{ma2023stablecoin} for a detailed discussion) or decentralized, allowing anyone to mint new tokens upon posting sufficient collateral. For further discussion on stablecoin classifications, see~\cite{hafner2023four, ito2020stablecoin, kahya2021reducing, kwon2023drives}.\\

Confronted with the high and growing market capitalization and daily trading volumes of stablecoins,\footnote{Daily trading volumes for USDT and USDC typically represent 10 to 30\% of their market capitalization.} the variety of mechanisms proposed, and the numerous crashes and temporary de-pegging events -- such as the death spiral that led to the crash of the algorithmic stablecoin TerraUSD in 2022 (see~\cite{liu2023anatomy} and~\cite{uhlig2022luna}), and the temporary de-pegging of USDC for several days in 2023 following the bankruptcy of Silicon Valley Bank, one of the banks holding Circle's cash reserves backing USDC (see~\cite{eichengreen2023stablecoin}) -- stablecoins have garnered significant attention from academics, national and international financial institutions,\footnote{The BIS paper~\cite{kosse2023will} studies 68 stablecoins and shows that none managed to maintain parity with its peg at all times.} and policymakers.\\

Many papers of the recent literature are focused on the risk of so-called stablecoin runs.\footnote{This risk echoes classical issues in finance, such as money market funds ``breaking the buck'' (see~\cite{anadu2023runs},~\cite{havrylchyk2023reglementer}, and~\cite{oefele2023stablecoins} for a general analogy with money market funds, although stablecoins also share properties with exchange-traded funds, see~\cite{ma2023stablecoin}), the system of eurodollars (see~\cite{aldasoro2023money}), and the traditional problem of bank runs if one regards stablecoin issuers as unregulated banks circulating private money in the form of stablecoins, rather than banknotes as in the past.} With a peg comes indeed the risk that the peg breaks and agents sell off their stablecoins en masse due to growing concerns about the system backing the stablecoin. Such a scenario can turn into a self-fulfilling prophecy, driven by both signaling and mechanical effects, as increased sales further widen the gap from the peg. Loss of confidence in a stablecoin can stem from the very mechanism designed to maintain the peg, as seen with algorithmic stablecoins (see, for example,~\cite{adams2022runs}). It can also be triggered by negative news about the collateral, as occurred during the short-lived de-pegging of USDC following the bankruptcy of Silicon Valley Bank (see above). The factors influencing the probability of stablecoin runs are in fact numerous. The effects of commitment mechanisms, collateral, and decentralization are examined in~\cite{d2022can}, while the nontrivial role of transparency in reserve composition is discussed in~\cite{ahmed2024public}. The concentration of arbitrageurs is analyzed in~\cite{ma2023stablecoin}, where a tradeoff is identified between stability around the peg and the risk of runs -- see also~\cite{bertsch2023stablecoins}. Additionally,~\cite{kwon2023drives} uses a game-theoretic framework to demonstrate how price equilibria are shaped by the underlying architecture of the coin.\\

Another topic of interest in the literature is the impact of the success of stablecoins on the broader financial sector.~\cite{barthelemy2021stablecoins} examines the effects of stablecoins on short-term funding markets, while~\cite{flannery2023rise} explores the impact of stablecoin adoption on banks' lending capacity. The effect on the banking sector as a function of its structure is discussed by~\cite{liao2022stablecoins}, and~\cite{bertsch2023stablecoins} addresses various consequences of large-scale stablecoin adoption.\\

Questions regarding compensation for the de-pegging risk borne by stablecoin holders are explored in~\cite{gorton2022leverage}, echoing the debate on regulating stablecoins as securities, which would allow for dividend payments (see~\cite{ma2023stablecoin}). Broader regulatory concerns, along with other legal issues, are addressed in~\cite{arner2020stablecoins},~\cite{bains2022regulating},~\cite{board2023high}, and~\cite{gorton2023taming}.\footnote{The Markets in Crypto-Assets (MiCA) Regulation, aimed at ensuring market integrity and financial stability, includes several points regarding stablecoins, which have been in effect since June 30, 2024.}\\

Many studies focus on factors that explain the average size (in absolute value) of peg deviations (see~\cite{lyons2023keeps} for a notable example). However, models for the dynamics of stablecoins remain relatively uncommon in the literature. Notable stochastic models include those presented in~\cite{klages2021stability},~\cite{klages2022while}, and~\cite{wang2020stablecoins}.\\

In this paper, we do not focus on the specific mechanisms underlying stablecoins. Instead, our emphasis is on modeling the price dynamics of one stablecoin relative to another, with the aim of developing an automated market-making algorithm for swapping two cryptos pegged to the same asset. The model we propose is based on stationary stochastic processes and goes beyond the use of a simple Ornstein-Uhlenbeck (OU) process. Specifically, we utilize multi-level nested OU processes, which are reminiscent of the two-factor Hull-White model used for interest rate term structures (see~\cite{hull1994numerical}). In this framework, the price process mean-reverts towards a secondary process, which itself mean-reverts towards a constant value. For pairs of stablecoins pegged to the same asset, this type of models is, of course, far more realistic than the nonstationary Brownian diffusion models typically used in automated market making for most pairs of (crypto)currencies (see~\cite{barzykin2023dealing} and~\cite{bergault2024automated}). It also improves upon the classical OU model by accommodating more complex dynamics, such as longer-lasting de-pegging phases for one or both stablecoins. The use of a multi-level nested OU process further strengthens the model by mitigating the often exaggerated and detrimental confidence in the mean reversion towards a specific value that is commonly seen in models solely based on a standard OU process.
\\

Our modeling framework also extends to other pairs of cryptos beyond stablecoins, in particular pairs involving a non-rebasing liquid staking token and the underlying native token. In Proof of Stake blockchains, staking involves locking native tokens within the network to support its operations and ensure its security, while earning rewards in return. Traditionally, once native tokens are staked, they become inaccessible, meaning holders cannot use them for other activities, such as participating in decentralized finance (DeFi). Liquid staking, however, allows users to receive a liquid staking token in exchange for the native tokens they stake, providing them with liquidity and enabling participation in other operations while still earning staking rewards. A major example on the Ethereum blockchain is the Lido stETH token, which has a market capitalization of 24 B\$. Other examples on Ethereum include the ether.fi Staked Ether token (eETH) and Mantle Staked Ether (mETH). Some liquid staking tokens also have a non-rebasing, value-accruing counterpart, such as wrapped stETH (wstETH) in the case of the Lido stETH token, which has a market capitalization of nearly 10 B\$. While these assets are not pegged to fiat currencies, they are pegged to the native cryptocurrency of the blockchain (ETH in the above examples), with an interest rate component for value-accruing tokens.\\

By using multi-level nested OU processes to model asset price dynamics, this paper introduces a novel Automated Market Maker (AMM) model specifically tailored for pegged assets. This model draws inspiration from both the recent literature on algorithmic market making in traditional finance and the authors' recent works~\cite{bergault2024automated, bergault2024price}, which proposed price-aware AMMs for pairs of volatile cryptos.\\

The quantitative literature on algorithmic market making began in the 1980s with the pioneering works of Ho and Stoll~\cite{ho1981optimal, ho1983dynamics} and gained renewed interest in 2008 with the influential work of Avellaneda and Stoikov~\cite{avellaneda2008high}, which has since become a key reference in quantitative finance. Building on this foundation,~\cite{gueant2013dealing} extended the analysis by providing closed-form approximations for optimal quotes under exponential intensity models. The expected utility framework of~\cite{avellaneda2008high} was later complemented by the widely used approach introduced in~\cite{cartea2014buy}, where the market maker maximizes expected PnL while penalizing inventory through a running quadratic cost.\footnote{The two frameworks have been extensively compared in~\cite{gueant2017optimal}.} Our approach builds on recent advances in the literature (see~\cite{barzykin2022market, barzykin2023algorithmic, barzykin2023dealing, bergault2021size}).\footnote{For a comprehensive review of market making literature, see the books~\cite{cartea2015algorithmic} and~\cite{gueant2016financial}.} The papers most related to ours include~\cite{barzykin2024algorithmic}, which introduces stationary processes in market making to model the cointegrated dynamics of spot and futures prices for precious metals, and our previous work~\cite{bergault2024automated}, which considers the specifics of automated market making platforms where liquidity providers evaluate excess gains or losses relative to a benchmark (in this case, the Hodl benchmark).\\

This paper is organized as follows. Section 2 presents the dynamic model we consider for the price of a stablecoin relative to another. We discuss parameter estimation and filtering techniques, and explore how the model can be extended to the price dynamics of a non-rebasing liquid staking token in terms of the native token. Section 3 introduces the stochastic optimal control problem used in this paper and demonstrates how to approximate solutions that can be used to design a price-aware automated market maker for pegged assets. Section 4 presents calibration on real-world data and numerical examples. The appendix contains technical results related to stochastic filtering.

\section{A Model for Stablecoin Exchange Rate Dynamics}

\subsection{Multi-level Nested Ornstein-Uhlenbeck Model}

Unlike many cryptos, which are characterized by high volatility and frequent price jumps, stablecoins are specifically designed to maintain stability relative to another asset, most commonly a fiat currency like the US dollar. In practice, popular stablecoins typically oscillate around a fixed peg, as shown in Figures~\ref{fig:USDC} and \ref{fig:USDT}, which display the price trajectories of USDC and USDT in US dollars, respectively. However, even when well-designed, stablecoins inevitably experience temporary periods of de-pegging, often due to a loss of confidence, which can be driven by endogenous factors within the crypto ecosystem or by exogenous events, such as new regulations or, to give a famous example, the bankruptcy of Silicon Valley Bank (SVB).\\

In contrast to many models used in the crypto sphere that primarily focus on capturing high volatility and jump risks, relevant models for the price dynamics of stablecoins must incorporate mean-reverting behavior. In this paper, we do not model the price of stablecoins relative to their associated fiat currencies but rather relative to other stablecoins pegged to the same fiat currency. Such an exchange rate naturally exhibits mean-reverting properties as well.\\

In continuous-time finance, the most widely used mean-reverting stochastic process is the OU process, represented by the following stochastic differential equation for the exchange rate process $(S_t)_{t}$:
\begin{equation*}
    dS_t = -\kappa (S_t - \bar{S}) \, dt + \sigma \, dW^S_t,
\end{equation*}
where \( \kappa \) represents the rate of mean reversion, \( \bar{S} \) is the long-term mean, which reflects the peg, \( \sigma \) is a volatility parameter,\footnote{The term volatility is somewhat misleading for a mean-reverting process, but we will use it throughout the paper.} and \( (W^S_t)_{t} \) is a standard Brownian motion. The issue with the standard OU model is that it often fails to precisely capture the sometimes long-lasting de-pegging phases observed in practice. Furthermore, the OU process tends to introduce an exaggerated confidence in reversion to a fixed constant, which can be detrimental, particularly in an optimal control framework. A solution, already applied in the fixed-income literature by Hull and White \cite{hull1994numerical}, is to relax the assumption of a fixed \( \bar{S} \) by assuming that the mean-reversion target is itself stochastic, but of course unobservable. In this article, we propose that the exchange rate process $(S_t)_{t}$ is governed by a multi-level nested OU process:
\begin{equation}
\begin{cases}
    dS_t = -\kappa \left( S_t - U_t \right) \, dt + \sigma \, dW^S_t, \\
    dU_t = -\eta \left( U_t - \bar{U} \right) \, dt + \nu \, dW^U_t,
\end{cases}
\label{eq:NOU_system}
\end{equation}
where \( \bar{U} \in \mathbb{R}_+^* \) is the long-term average of the two processes, \( \kappa > \eta > 0 \) are the mean reversion rates of the exchange rate towards the current value of the underlying process $(U_t)_{t}$ and of $(U_t)_{t}$ towards \( \bar{U} \), respectively. The parameters \( \sigma, \nu > 0 \) represent the volatility parameters of the processes, and \( \left( W^S_t, W^U_t \right)_{t} \) is a standard bidimensional Brownian motion.

\begin{figure}[!h]
    \centering
    \includegraphics[width=0.83\textwidth]{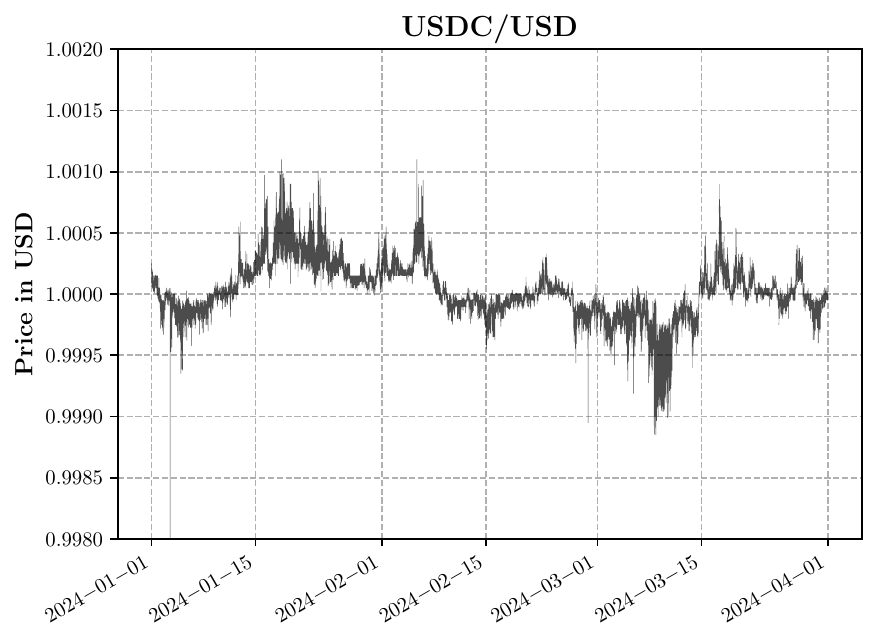}
    \caption{Price of USDC in US dollars between January 1, 2024, and March 31, 2024.}
    \label{fig:USDC}
\end{figure}

\begin{figure}[!h]
    \centering 
    \includegraphics[width=0.83\textwidth]{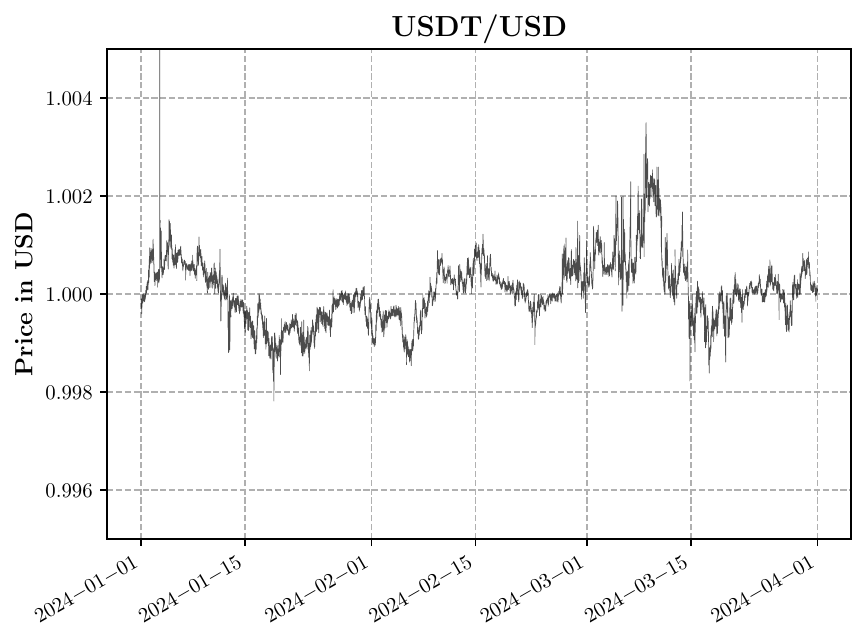}
    \caption{Price of USDT in US dollars between January 1, 2024, and March 31, 2024.}
    \label{fig:USDT}
\end{figure}

\subsection{Estimation of the Parameters}
\label{section:estim-nestedOU}

Using multi-level nested OU processes requires estimating the parameters \(\bar{U}\), \(\kappa\), \(\eta\), \(\sigma\), and \(\nu\) based on observations of the process \((S_t)_{t}\). In what follows, we demonstrate that \((S_t)_{t}\) is a Gaussian process characterized by known mean and covariance function. Consequently, we propose employing a maximum likelihood estimator (MLE) to infer the parameters of the model.\\

Let us start with a proposition stating the expression of $S_t$ as a function of initial values for $S$ and $U$:
\begin{prop}
\label{prop_NOU}
Given initial values \((S_0, U_0)\) at time \(t=0\), the process \((S_t)_{t }\) can be expressed for all \(t \ge 0\) as follows:
\begin{align*}
    S_t &= \bar{U} + (S_0 - \bar{U}) e^{-\kappa t} + \frac{\kappa}{\kappa - \eta} (U_0 - \bar{U}) \left( e^{-\eta t} - e^{-\kappa t} \right) \\
    &\quad + \frac{\kappa}{\kappa - \eta}\nu \int_0^t \left( e^{-\eta (t-s)} - e^{-\kappa (t-s)} \right) dW^U_s + \sigma \int_0^t e^{-\kappa (t-s)} dW^S_s.
\end{align*}
\end{prop}

\begin{proof}
The process \((U_t)_{t}\) follows a classical OU dynamic, given by
\[
U_t = \bar{U} + (U_0 - \bar{U}) e^{-\eta t} + \nu \int_0^t e^{-\eta (t - s)} \, dW^U_s.
\]
Similary, we know that the solution of the stochastic differential equation for \((S_t)_t\) is
\[
S_t = S_0 e^{-\kappa t}  +  \kappa \int_0^t e^{-\kappa(t- s)} U_s \, ds +  \sigma \int_0^t e^{-\kappa (t-s)} \, dW^S_s.
\]

Next, substituting the expression for \(U_s\) and applying Fubini's theorem, we get
\[
\kappa \int_0^t e^{-\kappa(t- s)} U_s \, ds = \bar{U} \left(1- e^{-\kappa t} \right) + \frac{\kappa}{\kappa - \eta} (U_0 - \bar{U}) \left( e^{-\eta t} - e^{- \kappa t} \right) + \frac{\kappa}{\kappa - \eta} \nu \int_0^t \left( e^{-\eta (t - s)} - e^{-\kappa (t-s)} \right) dW^U_s.
\]

Combining these results, we obtain
\begin{align*}
    S_t &= \bar{U} + (S_0 - \bar{U}) e^{-\kappa t} + \frac{\kappa}{\kappa - \eta} (U_0 - \bar{U}) \left( e^{-\eta t} - e^{-\kappa t} \right) \\
    &\quad + \frac{\kappa }{\kappa - \eta}\nu \int_0^t \left( e^{-\eta (t-s)} - e^{-\kappa (t-s)} \right) dW^U_s + \sigma \int_0^t e^{-\kappa (t-s)} dW^S_s.
\end{align*}
\end{proof}

As we do not observe the value of the underlying process \((U_t)_{t}\) at any point in time, it is preferable to consider the unconditional/stationary version of the process \((S_t)_{t}\), which is given in the corollary below:

\begin{cor}
A stationary representation of the process \((S_t)_{t}\) is:
\[
S_t = \bar{U} + \frac{\kappa}{\kappa - \eta} \nu \int_{-\infty}^t \left( e^{-\eta (t-s)} - e^{-\kappa (t-s)} \right) dW^U_s + \sigma \int_{-\infty}^t e^{-\kappa (t-s)} dW^S_s.
\]
In particular, the unconditional mean and covariance function are given by the following expressions:
$$\mathbb{E}[S_t] = \bar U,$$
$$C(t,s; \kappa, \eta, \sigma, \nu) = \text{Cov}(S_t, S_s) = \frac 12 \frac{\kappa^2}{\eta(\kappa^2 - \eta^2)} \nu^2 e^{-\eta|t-s|} + \frac 12 \left( \frac{\sigma^2}\kappa - \frac{\kappa}{(\kappa^2 - \eta^2)} \nu^2\right) e^{-\kappa|t-s|}.$$
\end{cor}

\begin{proof}
The expression of the unconditional/stationary version of $S_t$ follows from Proposition \ref{prop_NOU}. We straightforwardly deduce that the unconditional expected value of the process $(S_t)_t$ is $\bar U$.\\

As far as the covariance function \( (t,s) \mapsto \text{Cov}(S_t, S_s) \) is concerned, we have (assuming without loss of generality that~\( t \geq s \)):
\begin{eqnarray*}
    \text{Cov}(S_t, S_s) &=& \text{Cov}\left( \frac{\kappa}{\kappa - \eta} \nu \int_{-\infty}^t \left( e^{-\eta(t-u)} - e^{-\kappa(t-u)} \right) dW^U_u + \sigma \int_{-\infty}^t e^{-\kappa(t-u)} dW^S_u, \right. \\
    &&\qquad\quad \left. \frac{\kappa}{\kappa - \eta} \nu \int_{-\infty}^s \left( e^{-\eta(s-u)} - e^{-\kappa(s-u)} \right) dW^U_u + \sigma \int_{-\infty}^s e^{-\kappa(s-u)} dW^S_u \right).\\
    &=& \text{Cov}\left(\frac{\kappa}{\kappa - \eta} \nu \int_{-\infty}^t \left( e^{-\eta(t-u)} - e^{-\kappa(t-u)} \right) dW^U_u, \frac{\kappa}{\kappa - \eta} \nu \int_{-\infty}^s \left( e^{-\eta(s-u)} - e^{-\kappa(s-u)} \right) dW^U_u\right)\\
    && +\ \text{Cov}\left(\sigma \int_{-\infty}^t e^{-\kappa(t-u)} dW^S_u, \sigma \int_{-\infty}^s e^{-\kappa(s-u)} dW^S_u \right)\\
    &=& \left( \frac{\kappa }{\kappa - \eta}  \right)^2 \nu^2 \int_{-\infty}^{s} \left( e^{-\eta(t-u)} - e^{-\kappa(t-u)} \right) \left( e^{-\eta(s-u)} - e^{-\kappa(s-u)} \right) du\\
    && + \sigma^2 \int_{-\infty}^{s} e^{-\kappa(t-u)} e^{-\kappa(s-u)} du,
\end{eqnarray*}
where we used the independence between the two Brownian motions.\\

A straightfoward computation gives
\begin{eqnarray*}
&&\int_{-\infty}^{s} \left( e^{-\eta(t-u)} - e^{-\kappa(t-u)} \right) \left( e^{-\eta(s-u)} - e^{-\kappa(s-u)} \right) du\\
&=& \frac 1{2\eta} e^{-\eta(t-s)} + \frac 1{2\kappa} e^{-\kappa(t-s)} - \frac 1{\kappa+\eta} e^{-\kappa(t-s)} - \frac 1{\kappa+\eta} e^{-\eta(t-s)}.
\end{eqnarray*}
and $$\int_{-\infty}^{s} e^{-\kappa(t-u)} e^{-\kappa(s-u)} du = \frac1{2\kappa}  e^{-\kappa(t-s)}.$$

Combining the above and rearranging the terms, we get:
\[
\text{Cov}(S_t, S_s) = \frac{1}{2} \frac{\kappa^2}{\eta (\kappa^2 - \eta^2)} \nu^2 e^{-\eta (t-s)} + \frac{1}{2} \left( \frac{\sigma^2}{\kappa} - \frac{\kappa}{\kappa^2 - \eta^2} \nu^2 \right) e^{-\kappa (t-s)}.
\]
Exchanging the role of $t$ and $s$, we eventually get the result.\\
\end{proof}

Because $(S_t)_t$ is a Gaussian process with mean $\bar U$ and known covariance function $C(\cdot, \cdot; \kappa, \eta, \sigma, \nu)$, the parameters $\bar U, \kappa, \eta, \sigma$, and $\nu$ can be estimated in a classical way. Using a sample $\mathcal S = (S_{t_1}, \ldots, S_{t_d})'$ of observations at times $t_1 < \ldots < t_d$ in the form of a column vector, one can for example maximize over $(\bar U, \kappa, \eta, \sigma, \nu) \in \{v = (v_1, \ldots, v_5)  \in \mathbb R^5 | v_2>v_3>0, v_4 >0, v_5 > 0\}$ the log-likelihood of the sample which writes
\begin{eqnarray*}
\mathcal L\mathcal L(S_{t_1}, \ldots, S_{t_d} ; \bar U, \kappa, \eta, \sigma, \nu) &=& - \frac d2 \log(2\pi) - \frac 12 \log\left(\textrm{det}\left(C(t_i, t_j; \kappa, \eta, \sigma, \nu)\right)_{1\le i,j\le d}\right)\\
&& - \frac 12 (\mathcal S - \bar U {1})'\left(C(t_i, t_j; \kappa, \eta, \sigma, \nu)\right)_{1\le i,j\le d}^{-1} (\mathcal S - \bar U {1}),
\end{eqnarray*}
where ${1}$ is a $d$-dimensional column vector with all coordinates equal to $1$.\\

\begin{rem}
For a subdivision $(t_1, \ldots, t_d)$ with a fixed time step between consecutive observations, the covariance matrix $\left(C(t_i, t_j; \kappa, \eta, \sigma, \nu)\right)_{1\le i,j\le d}$ is a Toeplitz matrix.
\end{rem} 

It is worth noting that, in order to reduce the dimensionality of the problem from $5$ to $4$ and optimize over a rectangular set, we can use the fact that the MLE estimator of $\bar{U}$ is $\bar{\mathcal{S}} = \frac{1}{d} \sum_{1 \leq i \leq d} S_{t_i}$. The remaining parameters can then be estimated by numerically approximating the minimizer of the following function:
\begin{eqnarray*}
(\delta, \eta, \sigma, \nu) \in \left(\mathbb{R}_+^*\right)^4 &\mapsto & \log\left(\textrm{det}\left(C(t_i, t_j; \eta + \delta, \eta, \sigma, \nu)\right)_{1 \leq i,j \leq d}\right) \\
&&+ (\mathcal{S} - \bar{\mathcal{S}} \mathbf{1})'\left(C(t_i, t_j; \eta + \delta, \eta, \sigma, \nu)\right)_{1 \leq i,j \leq d}^{-1} (\mathcal{S} - \bar{\mathcal{S}} \mathbf{1}).
\end{eqnarray*}
and use the change of variables $\kappa = \eta + \delta$.

\subsection{Filtering Procedure}
\label{sec:filter-procedure}

Beyond estimating the parameters driving the dynamics of $(S_t)_t$, a significant challenge with multi-level nested OU processes is that the current mean-reversion target is not directly observable. As a result, it is ineffective to consider an optimal control framework based on the pair of state variables $(S,U)$ because optimal controls cannot depend on $U$, which is unobservable.\\

The classical results of stochastic filtering, recalled in the appendix, provide an alternative to the unobservable process $(U_t)_t$ by considering the process $(\widehat{U}_t)_t = \left(\mathbb{E}[ U_t | \mathcal{F}_t^S ]\right)_t$, where $\mathbb F^S = \left(\mathcal{F}_t^S\right)_t$ represents the natural filtration associated with $(S_t)_t$. More precisely, we have the following result:

\begin{prop}
Let $(S_t, U_t)_t$ be a solution of the stochastic differential equation \eqref{eq:NOU_system}, and let $\widehat{U}_t = \mathbb{E}[ U_t | \mathcal{F}_t^S ]$.\\
Then, the pair $(S_t, \widehat{U}_t)_t$ satisfies the system
\begin{equation*}
\begin{cases}
    dS_t = -\kappa \left( S_t - \widehat{U}_t \right) \, dt + \sigma \, d\widehat{W}^S_t, \\
    d\widehat{U}_t = -\eta \left( \widehat{U}_t - \bar{U} \right) \, dt + \frac{\kappa}{\sigma} V_t \, d\widehat{W}^S_t,
\end{cases}
\end{equation*}
where $\widehat{W}^S_t = W^S_t + \frac{\kappa}{\sigma} \int_0^t (U_s - \widehat{U}_s) ds$ defines a $\left(\mathcal{F}_t^S\right)_t$-Brownian motion, and $V_t = \mathbb{V}(U_t | \mathcal{F}_t^S)$ solves the ordinary differential equation
$$
\frac{dV_t}{dt} = -2\eta V_t + \nu^2 - \frac{\kappa^2}{\sigma^2}V_t^2.
$$
\end{prop}

In practice, an initial value for $\widehat{U}$ must be selected. One possible choice is $\bar{U}$. For the variance process $(V_t)_t$, an initial value can be chosen, or the asymptotic value
$$
V_{\infty} = \frac{\sigma^2}{\kappa^2} \left( -\eta + \sqrt{\eta^2 + \left(\frac{\kappa \nu}{\sigma}\right)^2}\right) = \frac{\nu^2}{\eta + \sqrt{\eta^2 + \left(\frac{\kappa \nu}{\sigma}\right)^2}}
$$
can be used, leading to the stochastic differential equation
\begin{equation}
\begin{cases}
    dS_t = -\kappa \left( S_t - \widehat{U}_t \right) \, dt + \sigma \, d\widehat{W}^S_t, \\
    d\widehat{U}_t = -\eta \left( \widehat{U}_t - \bar{U} \right) \, dt + \widehat{\nu} \, d\widehat{W}^S_t,
\end{cases}
\qquad \text{where} \quad \widehat{\nu} = \nu \frac{\frac{\kappa \nu}{\sigma}}{\eta + \sqrt{\eta^2 + \left(\frac{\kappa \nu}{\sigma}\right)^2}}.
\label{eq:NOU_system_filtered}
\end{equation}

Unlike with the pair $(S,U)$, it is now feasible to consider an optimal control framework based on the state variables $(S,\widehat{U})$.

\subsection{Extension to Liquid Staking Tokens}
\label{section:ext-lst}

Before introducing the optimal control framework that will allow use to  define an automated market-making strategy for swapping one stablecoin for another, let us highlight that the mathematical framework we proposed for modeling stablecoin exchange rate dynamics can easily be adapted -- up to a discounting factor -- to pairs of cryptocurrencies consisting of a non-rebasing liquid staking token and its underlying native token. If $(\tilde{S}_t)_t$ represents the price of the non-rebasing liquid staking token in terms of the underlying native token, it is natural to assume that the rebased price process $(S_t)_t = \left(\tilde{S}_t \exp\left(-\int_0^t r_s ds\right)\right)_t$ follows the same dynamics as that of a stablecoin, where $(r_t)_t$ denotes the yield associated with staking.\\

In practice, one could rebase the price using a time series of yields $(r_t)_t$, but this raises the issue of forecasting that yield for control purposes. Alternatively, we can assume that the yield process remains approximately constant over short time scales and estimate a constant yield $r$ from the recent realized drift of log prices. Section \ref{sec:empirical} will present some real data of liquid staking tokens which clearly exhibits stablecoin-like dynamics once discounted.\\

As a result of the above, the algorithms developed in this paper can be applied to both types of crypto pairs: directly for stablecoin-stablecoin pairs, and, with the addition of a discounting factor, for pairs involving a non-rebasing liquid staking token and its underlying native token.\\

\section{Designing a Price-Aware Automated Market Maker for Pegged Assets}
\label{sec:optimal-model}

\subsection{The Optimal Control Problem}
\label{sec:optimal-control}

We consider a liquidity pool composed of two pegged cryptos (e.g., a USDC/USDT pool), referred to as crypto \( 0 \) and crypto \( 1 \), respectively. An exogenous market exchange rate, representing the price of crypto \( 1 \) in terms of crypto \( 0 \), is assumed to be observable and is modeled as a stochastic process \( (S_t)_t \). This exchange rate could, for instance, correspond to the mid-price on a centralized exchange with a limit order book, such as Binance, Kraken, or Coinbase, or it could be provided by an oracle such as Pyth. Importantly, this rate is indicative and not directly tradable, even for infinitesimally small transactions. We model the process \( (S_t)_t \) as following the dynamics of a multi-level nested Ornstein-Uhlenbeck (OU) process, defined on a probability space \( (\Omega, \mathbb{F}, \mathbb{P}) \), with parameters specified in Eq. \eqref{eq:NOU_system}.\\

For a given time horizon \( T > 0 \), we assume the automated market maker (AMM) provides liquidity to traders for swapping the two cryptos. Let \( S^{1,0}(t,z) = S_t - \delta^{1,0}(t,z) \) denote the exchange rate at which the AMM agrees to buy a quantity \( z \) of crypto \( 1 \) in exchange for crypto \( 0 \) at time \( t \in [0,T] \), and let \( S^{0,1}(t,z) = S_t + \delta^{0,1}(t,z) \) denote the exchange rate at which the AMM agrees to sell a quantity \( z \) of crypto \( 1 \) in exchange for crypto \( 0 \) at time \( t \in [0,T] \).\\

To simplify the analysis, we follow the approach used in \cite{bergault2024automated}, originally introduced in \cite{barzykin2023dealing}, and assume that the markups accumulate separately from the reserves of cryptos \( 0 \) and \( 1 \) held in the pool. We model the accumulated markups using a process \( (X_t)_{t} \), whose dynamics are given by
\begin{equation}
dX_t = \int_{z \in \mathbb{R}_+^*} z\delta^{0,1}(t,z) J^{0,1}(dt,dz) + \int_{z \in \mathbb{R}_+^*} z\delta^{1,0}(t,z) J^{1,0}(dt,dz),
\label{Xdyn}
\end{equation}
where \( J^{0,1}(dt,dz) \) and \( J^{1,0}(dt,dz) \) are two \( \mathbb{R}_+^* \)-marked point processes representing transactions. Here, \( J^{0,1}(dt,dz) \) models instances where the AMM sells crypto \( 1 \) and receives crypto \( 0 \), while \( J^{1,0}(dt,dz) \) represents the AMM selling crypto \( 0 \) and receiving crypto \( 1 \).\\

These marked point processes also allow to write the dynamics of the pool reserves in crypto $0$ (represented by the process $(q^0_t)_t$) and crypto $1$ (represented by the process $(q^1_t)_t$):
$$dq^0_t =   -S_t \left( \int_{z \in \mathbb R_+^*} z \left(J^{1,0}(dt,dz) -  J^{0,1}(dt,dz) \right)\right) \quad \text{and} \quad dq^1_t =  \int_{z \in \mathbb R_+^*} z \left(J^{1,0}(dt,dz) -  J^{0,1}(dt,dz) \right).$$

As in~\cite{bergault2024automated}, we assume that \( J^{0,1}(dt,dz) \) and \( J^{1,0}(dt,dz) \) have known intensity kernels, denoted by \( (\nu^{0,1}_t(dz))_t \) and \( (\nu^{1,0}_t(dz))_t \), which satisfy
$$
\nu^{0,1}_t(dz) = \Lambda^{0,1}\left(z, \delta^{0,1}(t,z)\right)m(dz) \quad \text{and} \quad \nu^{1,0}_t(dz) = \Lambda^{1,0}\left(z, \delta^{1,0}(t,z)\right)m(dz),
$$
where \( m \) is a measure (typically Lebesgue or discrete), and \( \Lambda^{0,1} \) and \( \Lambda^{1,0} \) are intensity functions defined as
$$
\Lambda^{0,1}(z,\delta) = \lambda^{0,1}(z) \frac{1}{1 + e^{a^{0,1}(z) + b^{0,1}(z) \delta}} \quad \text{and} \quad \Lambda^{1,0}(z,\delta) = \lambda^{1,0}(z) \frac{1}{1 + e^{a^{1,0}(z) + b^{1,0}(z) \delta}}.
$$
Informally, \( \lambda^{0,1}(z) m(dz) \) and \( \lambda^{1,0}(z) m(dz) \) represent the maximum number (per unit of time) of transactions of size within the infinitesimal interval \( [z, z+dz] \), or equivalently, the height of the demand curve. The parameters \( a^{0,1}(z) \), \( b^{0,1}(z) \), \( a^{1,0}(z) \), and \( b^{1,0}(z) \) describe the sensitivity of liquidity takers to markups, thereby shaping the demand curve.\\

In what follows, we focus on the PnL of an agent who deposits tokens into the pool's reserves (referred to as a liquidity provider) in comparison to an agent who holds the tokens outside the AMM.\footnote{This is commonly referred to as the Hodl benchmark in the AMM literature.} We therefore introduce the following two processes:
\begin{equation}
\left(Y^0_t \right)_{t} = \left(q^0_t - q^0_0 \right)_{t \in \mathbb{R}_+} \quad \text{and} \quad \left(Y^1_t \right)_{t} = \left(q^1_t - q^1_0 \right)_{t}.
\label{Ydyn}
\end{equation}

The mark-to-market excess PnL at time $T$ is given by
\begin{eqnarray*}
&&X_T + Y^0_T + Y^1_T S_T \\
&=& \int_0^T \left\{ \int_{z \in \mathbb{R}_+^*} z \delta^{0,1}(t,z) J^{0,1}(dt,dz) + \int_{z \in \mathbb{R}_+^*} z \delta^{1,0}(t,z) J^{1,0}(dt,dz) - Y^1_t \kappa (S_t - U_t) dt + \sigma Y^1_t dW^S_t \right\}\\
&=& \int_0^T \left\{ \int_{z \in \mathbb{R}_+^*} z \delta^{0,1}(t,z) J^{0,1}(dt,dz) + \int_{z \in \mathbb{R}_+^*} z \delta^{1,0}(t,z) J^{1,0}(dt,dz) - Y^1_t \kappa (S_t - \widehat U_t) dt + \sigma Y^1_t d\widehat W^S_t \right\}.
\end{eqnarray*}

The approach we propose for selecting the markups is based on several principles. First, the markups must depend solely on observable quantities. In other words, the markups at time \( t \) can only rely on information available prior to \( t \), with dependence on price dynamics occurring through \( S \) and \( \widehat{U} \), not through \( U \), which is not observable. Second, the markups should enable the liquidity provider to earn profits from their liquidity provision and by taking advantage of the mean-reverting nature of the exchange rate while also limiting risk exposure relative to the Hodl benchmark. Given the exchange rate dynamics, maximizing the expected excess PnL while penalizing deviations from the Hodl benchmark, as in \cite{bergault2024automated}, is insufficient. This is because the risk of deviating from Hodl depends on the price levels, because of mean reversion. Instead, we revert to the original objective function proposed in \cite{avellaneda2008high}, aiming to maximize the expected exponential (CARA) utility of the excess PnL at time \( T \):
\begin{align*}
\underset{\delta \in \mathcal{A}}{\sup} \mathbb{E} \left[-e^{-\gamma \left(X_T + Y^0_T + Y^1_T S_T \right)} \right]
\end{align*}
where \( \gamma \) is a risk aversion parameter, offering flexibility in strategy selection, and
\begin{equation}
\begin{split}
\mathcal{A}:= \left\lbrace \delta = \left(\delta^{0,1}, \delta^{1,0}\right) : \Omega \times [0,T] \times \mathbb{R}_+^* \mapsto \mathbb{R}^2 \mid \delta \text{ is } \mathcal{P}^{X,Y^0,Y^1,S} \otimes \mathcal{B}(\mathbb{R}_+^*) \text{-measurable} \right\rbrace,
\end{split}
\end{equation}
where \( \mathcal{P}^{X,Y^0,Y^1,S} \) denotes the \( \sigma \)-algebra of \( \mathbb{F}^{X,Y^0,Y^1,S} \)-predictable subsets of \( \Omega \times [0,T] \), and \( \mathcal{B}(\mathbb{R}_+^*) \) represents the Borel sets of \( \mathbb{R}_+^* \).\\

\begin{rem}
Unlike in \cite{bergault2024automated}, we do not account for the risk of reserve depletion on either side of the pool. Indeed, when the pool is sufficiently large or the risk aversion is not too low, this risk can be safely ignored.
\end{rem}

\subsection{From Hamilton-Jacobi-Bellman to Riccati: the quadratic Hamiltonian approximation}

The state variables associated with the optimal control problem correspond to the stochastic processes $(X_t)_t$, $(Y^0_t)_t$, $(Y^1_t)_t$, $(S_t)_t$, and $(\widehat{U}_t)_t$, whose dynamics are governed by the stochastic differential equations \eqref{eq:NOU_system_filtered}, \eqref{Xdyn}, and \eqref{Ydyn}.\\

The Hamilton-Jacobi-Bellman equation associated with this optimal control problem is therefore
\begin{eqnarray}
0 &=& \partial_t u(t,x,y^0, y^1, S,\widehat U) - \kappa  \left(S  - \widehat U \right) \partial_{S}u(t,x,y^0, y^1, S,\widehat U) - \eta \left(\widehat U-\bar U \right) \partial_{\widehat U} u(t,x,y^0, y^1, S,\widehat U)\nonumber\\
&& + \frac 12 \sigma^2 \partial_{SS}^2 u(t,x,y^0, y^1, S,\widehat U)+ \frac 12 \widehat\nu^2 \partial_{\widehat U\widehat U}^2 u(t,x,y^0, y^1, S,\widehat U) + \sigma \widehat\nu \partial_{S\widehat U}^2 u(t,x,y^0, y^1, S,\widehat U) \nonumber\\
&& + \int_{z \in \mathbb R_+^*} \underset{\delta}{\sup} \Lambda^{0,1}(z,\delta) \left( u(t,x+z\delta, y^0+zS, y^1-z, S, \widehat U) - u(t,x,y^0, y^1, S,\widehat U) \right)  m^{0,1}(dz) \nonumber \\
&& + \int_{z \in \mathbb R_+^*} \underset{\delta}{\sup} \Lambda^{1,0}(z,\delta) \left( u(t,x+z\delta, y^0-zS, y^1+z, S, \widehat U) - u(t,x,y^0, y^1, S,\widehat U) \right)  m^{1,0}(dz), \label{HJBu}
\end{eqnarray}
with terminal condition $u(T,x,y^0, y^1, S,\widehat U) = -e^{-\gamma \left(x+y^0 + y^1S\right)}$.\\

Using an ansatz inspired by that of \cite{gueant2017optimal}, here
$$u(t,x,y^0, y^1, S,\widehat U) = -e^{-\gamma \left( x+y^0 + y^1S+ \theta(t,y^1,S,\widehat U) \right)},$$ we see that Eq. \eqref{HJBu} can be simplified into
\begin{eqnarray}
0 &=& \partial_t \theta(t,y, S, \widehat U) - \kappa  \left(S  - \widehat U \right)  \left(y + \partial_S \theta(t,y, S, \widehat U)\right) -\eta \left(\widehat U-\bar U \right) \partial_{\widehat U} \theta(t,y, S, \widehat U)\nonumber\\
&& + \frac 12 \sigma^2 \partial_{SS}^2 \theta(t,y, S, \widehat U)+ \frac 12 \nu^2 \partial_{\widehat U\widehat U}^2 \theta(t,y, S, \widehat U) + \sigma \nu \partial_{\widehat US}^2 \theta(t,y, S, \widehat U) \nonumber\\
&&- \frac \gamma 2 \left(\sigma \left(y+ \partial_S \theta(t,y, S, \widehat U) \right) + \widehat\nu \partial_{\widehat U} \theta(t,y, S, \widehat U) \right)^2 \nonumber \\
&&  + \int_{z \in \mathbb R_+^*} zH^{0,1} \left(z,\frac{\theta(t,y, S, \widehat U) -  \theta(t,y-z, S, \widehat U) }{z}\right)m^{0,1}(dz) \nonumber\\
&&  + \int_{z \in \mathbb R_+^*} zH^{1,0}\left(z,\frac{\theta(t,y,S,\widehat U) -  \theta(t,y+z,S,\widehat U) }{z}\right)m^{1,0}(dz), \label{HJB}
\end{eqnarray}
with terminal condition $\theta(T,y,S,\widehat U) = 0$, where
$$
H^{0,1}:(z,p)\in\mathbb R_+^* \times \mathbb{R} \mapsto \underset{\delta}{\sup}\ \frac{\Lambda^{0,1}(z,\delta)}{\gamma z}(1 - e^{-\gamma z(\delta-p)})$$
and 
$$H^{1,0}:(z,p)\in\mathbb R_+^* \times \mathbb{R} \mapsto \underset{\delta}{\sup}\ \frac{\Lambda^{1,0}(z,\delta)}{\gamma z}(1 - e^{-\gamma z(\delta-p)}).$$

Under classical assumptions on the intensities (see \cite{gueant2017optimal}), we can prove, using techniques that have been repeatedly employed in the market-making literature, that given a smooth solution to Eq. \eqref{HJB}, the optimal markups are given by
\begin{equation}
\begin{aligned}\delta^{0,1*}(t,z) &= \bar\delta^{0,1} \left( z, \frac{\theta(t,Y^1_{t-},S_t,\widehat U_t) -  \theta(t,Y^1_{t-}-z,S_t,\widehat U_t) }{z} \right),\\
\delta^{1,0*}(t,z) &= \bar \delta^{1,0} \left(z, \frac{\theta(t,Y^1_{t-},S_t,\widehat U_t) -  \theta(t,Y^1_{t-}+z,S_t,\widehat U_t) }{z} \right),
\end{aligned}
\label{opt_quotes}
\end{equation}
with 
$$\bar \delta^{i,j}(z,p) = (\Lambda^{i,j})^{-1} \left(z, \gamma z H^{i,j}(z,p)-\partial_p{H^{i,j}} (z,p)  \right)$$
where for all $z$, $(\Lambda^{i,j})^{-1}(z, .)$ denotes the inverse of the function $\Lambda^{i,j}(z,.)$.\\

Solving Eq. \eqref{HJB} numerically on a grid presents several challenges. First, the equation is four-dimensional (time plus three state variables). Second, the presence of non-local terms makes traditional methods computationally intensive. Third, since the same Brownian motion drives both $S$ and $\widehat{U}$, nontrivial geometric considerations arise. However, it turns out that our choice of a multi-level nested OU process for $(S_t)_t$ makes it possible to apply the quadratic Hamiltonian trick introduced in \cite{evangelista2020closed} to approximate the solution~$\theta$ to Eq.~\eqref{HJB}.\\

Following~\cite{evangelista2020closed}, we can approximate the functions $H^{0,1}$ and $H^{1,0}$ with quadratic functions as follows:
$$
\check{H}^{0,1}(z,p) = \alpha^{0,1}_0(z) + \alpha^{0,1}_1(z) p + \frac{1}{2} \alpha^{0,1}_2(z) p^2, \qquad \check{H}^{1,0}(z,p) = \alpha^{1,0}_0(z) + \alpha^{1,0}_1(z) p + \frac{1}{2} \alpha^{1,0}_2(z) p^2.
$$

This leads to a new equation:
\begin{eqnarray}
0 &=& \partial_t \check{\theta}(t,y, S, \widehat{U}) - \kappa \left(S - \widehat{U} \right) \left(y + \partial_S \check{\theta}(t,y, S, \widehat{U})\right) - \eta \left(\widehat{U} - \bar{U} \right) \partial_{\widehat{U}} \check{\theta}(t,y, S, \widehat{U}) \nonumber\\
&& + \frac{1}{2} \sigma^2 \partial_{SS}^2 \check{\theta}(t,y, S, \widehat{U}) + \frac{1}{2} \nu^2 \partial_{\widehat{U}\widehat{U}}^2 \check{\theta}(t,y, S, \widehat{U}) + \sigma \nu \partial_{\widehat{U}S}^2 \check{\theta}(t,y, S, \widehat{U}) \nonumber\\
&& - \frac{\gamma}{2} \left(\sigma \left(y + \partial_S \check{\theta}(t,y, S, \widehat{U})\right) + \widehat{\nu} \partial_{\widehat{U}} \check{\theta}(t,y, S, \widehat{U}) \right)^2 \nonumber\\
&& + \int_{z \in \mathbb{R}_+^*} z \check{H}^{0,1} \left(z, \frac{\check{\theta}(t,y, S, \widehat{U}) - \check{\theta}(t,y-z, S, \widehat{U})}{z}\right) m(dz) \nonumber\\
&& + \int_{z \in \mathbb{R}_+^*} z \check{H}^{1,0} \left(z, \frac{\check{\theta}(t,y, S, \widehat{U}) - \check{\theta}(t,y+z, S, \widehat{U})}{z}\right) m(dz), \label{HJBquad}
\end{eqnarray}
with the terminal condition $\check{\theta}(T,y,S,\widehat{U}) = 0$. The solution to this equation will serve as our approximation of $\theta$.\\

The interest of the above approximation is that the solution to Eq. \eqref{HJBquad} is a polynomial of degree $2$ in $(y, S, \widehat U)$ with time-dependent coefficients. Let us make indeed the following ansatz:
$$\check \theta(t,y,S,\widehat U) = - \begin{pmatrix}
   y\\ S\\ \widehat U
\end{pmatrix}^\intercal A(t) \begin{pmatrix}
   y\\ S\\ \widehat U 
\end{pmatrix} - \begin{pmatrix}
   y\\ S\\ \widehat U
\end{pmatrix}^\intercal B(t) - C(t),$$
where $A:[0,T] \mapsto \mathcal M_3(\mathbb R)$, $B:[0,T] \mapsto \mathbb R^3$, and $C:[0,T] \mapsto \mathbb R$ are differentiable functions such that $A(T) = B(T) = C(T) = 0.$ Plugging this expression into Eq. \eqref{HJBquad} yields a system of ODEs for $A$, $B$, and $C$ that can easily be solved numerically:\footnote{We only report here the equations for $A$ and $B$, as $C$ is irrelevant for the computation of the optimal strategy.}

\begin{align*}
\begin{cases}
    A'(t) &= A(t)M^A A(t) + A(t)U^A + {U^A}^\intercal A(t) +R^A\\
    B'(t) &=  A(t) M^A B(t) + A(t) V^B + 2\Delta_{2,2,-1}{e_y}^\intercal A(t)e_y A(t)e_y + {U^A}^\intercal B(t)
\end{cases}
\end{align*}
where the matrix terms for the first equation are $M^A = 2\begin{pmatrix}
     \Delta_{2,1,1}& 0 & 0 \\
    0 & -\gamma\sigma^2 & -\gamma\sigma\widehat \nu \\
    0 & -\gamma\sigma\widehat \nu & -\gamma\widehat \nu^2  \\
\end{pmatrix}$, $U^A = \begin{pmatrix}
    0 & 0 & 0 \\
    \gamma \sigma^2 & \kappa & -\kappa \\
    \gamma \sigma \widehat \nu  & 0 & \eta \\
\end{pmatrix}$, and $R^A = \frac 12 \begin{pmatrix}
     -\gamma \sigma & -\kappa & \kappa \\
     -\kappa & 0 & 0 \\
     \kappa & 0 & 0 \\
\end{pmatrix}$ and the vector term for the second equation is $V^B = 
2 \begin{pmatrix}
\Delta_{1,1,-1}\\
0\\
-\eta \bar U  
\end{pmatrix},
$
with $e_y$ the first vector of the canonical basis $\begin{pmatrix}
1\\
0\\
0  
\end{pmatrix}$ and $\Delta_{i,j,\epsilon} = \int_0^{+\infty} \alpha^{1,0}_i(z)z^jm^{1,0}(dz)+\epsilon \int_0^{+\infty} \alpha^{0,1}_i(z)z^jm^{0,1}(dz).$\\

Computing numerically $A$ and $B$ raises no difficulty. Using $\check \theta$ instead of $\theta$ in Eq. \eqref{opt_quotes} corresponds to choosing the greedy policy with respect to $\check \theta$ in the reinforcement learning terminology and leads to the following markups:
\begin{equation}
\begin{aligned}\check\delta^{0,1*}(t,z) &= \bar\delta^{0,1} \left( z, \frac{\check \theta(t,Y^1_{t-},S_t,\widehat U_t) -  \check \theta(t,Y^1_{t-}-z,S_t,\widehat U_t) }{z} \right)\\
& = \bar\delta^{0,1} \left( z, \frac{ -2\begin{pmatrix}
   Y^1_{t-}\\ S_t\\ \widehat U_t
\end{pmatrix}^\intercal A(t) \begin{pmatrix}
   z\\ S_t\\ \widehat U_t 
\end{pmatrix} +  \begin{pmatrix}
   z\\ S_t\\ \widehat U_t
\end{pmatrix}^\intercal A(t) \begin{pmatrix}
   z\\ S_t\\ \widehat U_t 
\end{pmatrix} - z {e_y}^\intercal B(t) }{z} \right) ,\\
\check \delta^{1,0*}(t,z) &= \bar \delta^{1,0} \left(z, \frac{\check \theta(t,Y^1_{t-},S_t,\widehat U_t) -  \check \theta(t,Y^1_{t-}+z,S_t,\widehat U_t) }{z} \right),\\
& = \bar\delta^{1,0} \left( z, \frac{ 2\begin{pmatrix}
   Y^1_{t-}\\ S_t\\ \widehat U_t
\end{pmatrix}^\intercal A(t) \begin{pmatrix}
   z\\ S_t\\ \widehat U_t 
\end{pmatrix} +  \begin{pmatrix}
   z\\ S_t\\ \widehat U_t
\end{pmatrix}^\intercal A(t) \begin{pmatrix}
   z\\ S_t\\ \widehat U_t 
\end{pmatrix} +  z {e_y}^\intercal B(t) }{z} \right).
\end{aligned}
\label{markups}
\end{equation}

\begin{rem}
If the time horizon \( T \) is large, or if one explicitly wishes to consider the ergodic (limit) version of the problem -- i.e., the limit as \( T \to +\infty \), where \( t \) becomes irrelevant -- one should use the asymptotic markups and replace \( A(t) \) and \( B(t) \) with \( A(0) \) and \( B(0) \) in Eq. \eqref{markups}.
\end{rem}

\section{Empirical Analysis and Numerical Examples}
\label{sec:empirical}

In this section, we illustrate our approach using two pairs of pegged cryptos: the pair of stablecoins USDC/USDT and the pair wstETH/WETH, which consists of the non-rebasing liquid staking token wstETH and the wrapped token WETH, which is an ERC-20 compliant representation of the underlying native token~(ETH).\\

We begin by presenting the data and the results of the parameter estimation method discussed in Section~\ref{section:estim-nestedOU}. Next, we show the outcomes of the filtering techniques presented in Section~\ref{sec:filter-procedure} and explore different perspectives on the parameters driving the underlying process \( (U_t)_t \). Finally, we provide numerical simulations of the AMM strategy developed in Section~\ref{sec:optimal-model} (hereafter referred to as {NOU AMM}), first on simulated prices following multi-level nested OU dynamics to approximate the efficient frontier \textit{à la} Markowitz, and then using real prices. In the numerical simulations, we compare the outcomes of the {NOU AMM} strategy to those of the optimal strategy assuming the price follows a geometric Brownian motion. This strategy, described in detail in \cite{bergault2024automated}, is referred to as {GBM AMM} in this paper.

\subsection{Data}

To obtain indicative prices for the USDC/USDT and wstETH/WETH pairs, we first collected data for the following pairs: USDT/USD, USDC/USD, WETH/USD, and wstETH/USD. These dollar-denominated time series were retrieved from Pyth's Benchmarks API\footnote{See https://docs.pyth.network/benchmarks.} at a per-second frequency, with forward filling applied to handle missing data, covering the 3-month period from January 1, 2024, to March 31, 2024. Using these four time series, we then derived indicative price series for USDC/USDT and wstETH/WETH.\\

Figures \ref{fig:pyth-raw-data1} and \ref{fig:pyth-raw-data2} display the results of these computations. The mean-reverting characteristics of the stablecoin pair USDC/USDT, which motivated our multi-level nested OU model, are clearly visible. For the wstETH/WETH pair, we observe a trend around which the series exhibits rapid mean-reverting behavior.

\begin{figure}[!h]
    \centering
    \includegraphics[width=0.8\textwidth]{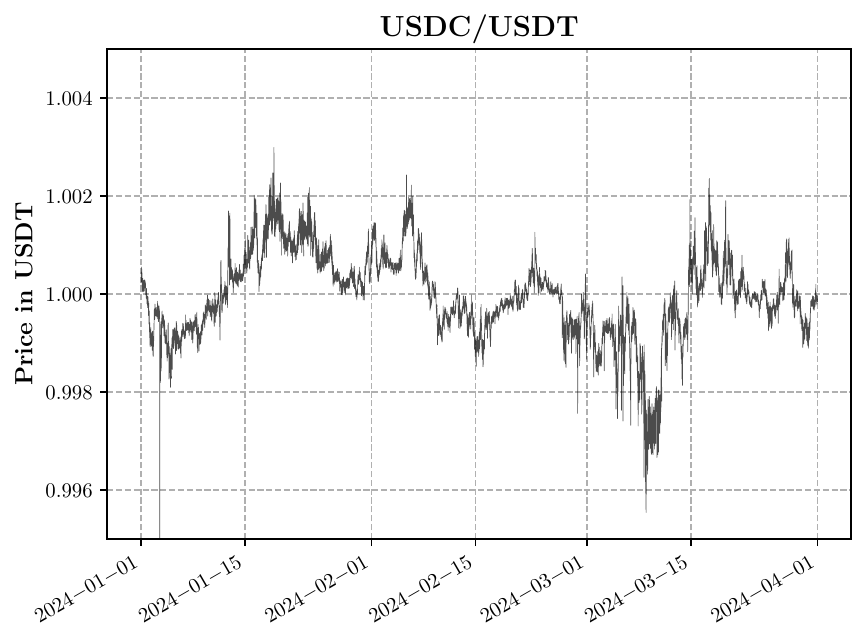}
    \caption{Price of USDC/USDT constructed from dollar-denominated Pyth data between January 1, 2024, and March 31, 2024.}
    \label{fig:pyth-raw-data1}
\end{figure}

\begin{figure}[!h]
    \centering
    \includegraphics[width=0.8\textwidth]{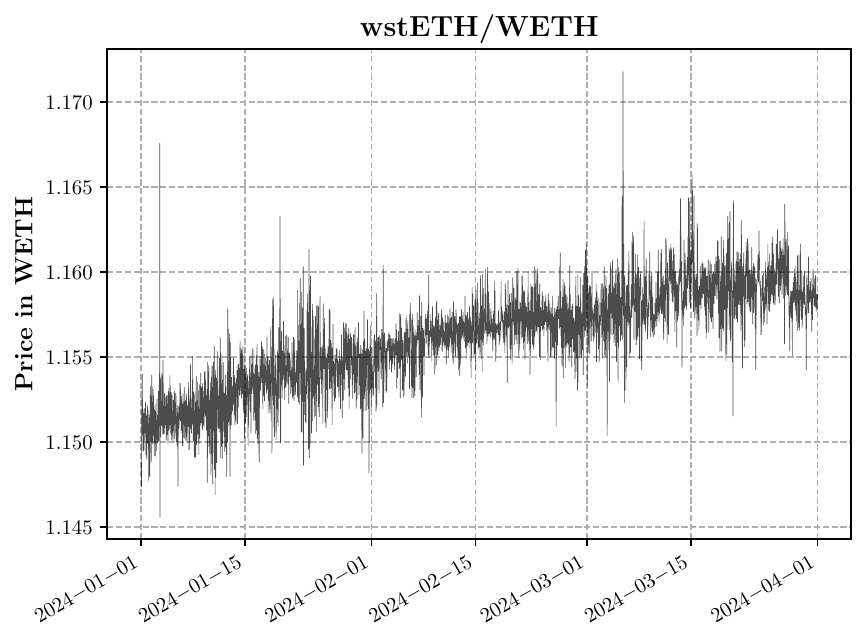}
    \caption{Price of wstETH/WETH constructed from dollar-denominated Pyth data between January 1, 2024, and March 31, 2024.}
    \label{fig:pyth-raw-data2}
\end{figure}

\subsection{Empirical Estimation of the Parameters}
\label{sec:calib-nestedOU}

Following Section \ref{section:estim-nestedOU}, we shall estimate the parameters of the multi-level nested OU model for the two pairs USDC/USDT and wstETH/WETH using a maximum likelihood procedure. For the wstETH/WETH pair, however, prices must first be adjusted to account for the yield generated by the staking activity, as described in Section \ref{section:ext-lst}.\\

To statistically isolate the component associated with the staking yield, we estimated an average yield by regressing the log-prices on time. This yielded an annualized rate \(r = 2.94\%\), consistent with observed staking yields during the period. We then discounted the original time series to obtain the adjusted series used in the maximum likelihood estimation. Figure \ref{fig:data-detrended-staked-assets} displays both the original and discounted time series.

\begin{figure}[!h]
    \centering
    \includegraphics[width=0.8\textwidth]{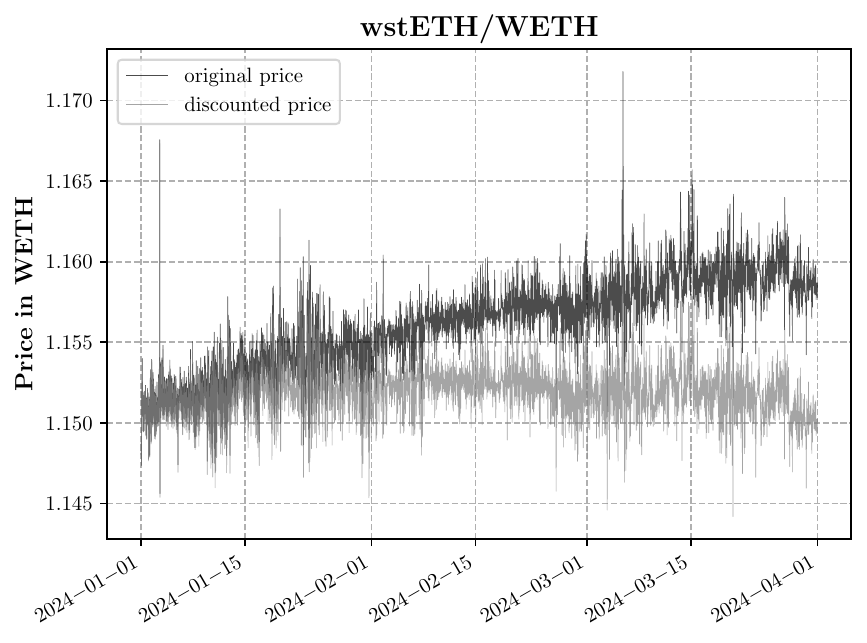}
    \caption{Price of wstETH/WETH between January 1, 2024, and March 31, 2024. The dark line represents the original price, while the light line shows the discounted price.}
    \label{fig:data-detrended-staked-assets}
\end{figure}

Tables \ref{table:params-stable} and \ref{table:params-LST} present the results of the parameter estimation for each month within the 3-month data interval. The maximum log-likelihood procedure described in Section \ref{section:estim-nestedOU} was applied to the time series resampled at a 15-minute frequency to mitigate microstructural noise.\\

\begin{table}[h!]
\renewcommand{\arraystretch}{1.2}
\begin{center}
\begin{tabular}{c||c|c||c|c|c}
&\multicolumn{2}{c||}{1st-level OU} & \multicolumn{3}{c}{2nd-level OU}\\
\hline
Estimation Period&Mean reversion&Volatility&Mean reversion&Volatility&Target\\
\hline
&$\kappa$ &$\sigma$ &$\eta$ &$\nu$ &$\bar{U}$\\
\hline
01-2024 & $4.57\times10^{-2}$ & $5.23\times10^{-4}$ & $1.6\times10^{-2}$ & $5.00\times10^{-4}$ & $1.0003$\\

02-2024 & $5.46\times10^{-2}$ & $5.22\times10^{-4}$ & $5.29\times10^{-2}$ & $5.45\times10^{-5}$ & $1.0000$\\%

03-2024 & $4.92\times10^{-2}$ & $5.21\times10^{-4}$ & $2.96\times10^{-2}$ & $4.77\times10^{-4}$ & $0.9996$\\%

\end{tabular}
\captionsetup{width=0.95\textwidth}
\caption{Estimated parameters for the price dynamics of USDC/USDT. Units are day$^{-1}$ for the mean reversion parameters $\kappa$ and $\eta$, USDT $\cdot$ day$^{-\frac 12}$ for the volatility parameters $\sigma$ and $\nu$, and USDT for $\bar U$.}
\label{table:params-stable}
\end{center}
\end{table}

\begin{table}[h!]
\renewcommand{\arraystretch}{1.2}
\begin{center}
\begin{tabular}{c||c|c||c|c|c}
&\multicolumn{2}{c||}{1st-level OU} & \multicolumn{3}{c}{2nd-level OU}\\
\hline
Estimation Period&Mean reversion&Volatility&Mean reversion&Volatility&Target\\
\hline
&$\kappa$ &$\sigma$ &$\eta$ &$\nu$ &$\bar{U}$\\
\hline
01-2024 & $10.09$ & $4.99\times10^{-3}$ & $4.90$ & $4.65\times10^{-3}$ & $1.1510$\\%

02-2024 & $6.11$ & $5.04\times10^{-3}$ & $2.14$ & $2.49\times10^{-3}$ & $1.1553$\\%

03-2024 & $8.82$ & $4.99\times10^{-3}$ & $4.78$ & $4.89\times10^{-3}$ & $1.1579$\\%

\end{tabular}
\captionsetup{width=0.95\textwidth}
\caption{Estimated parameters for the price dynamics of wstETH/WETH. Units are day$^{-1}$ for the mean reversion parameters $\kappa$ and $\eta$, WETH $\cdot$ day$^{-\frac 12}$ for the volatility parameters $\sigma$ and $\nu$, and WETH for $\bar U$.}
\label{table:params-LST}
\end{center}
\end{table}

Consistent with our observations, the estimated parameters reveal significant differences in the mean-reverting properties of the two time series. The stablecoin pair exhibits much slower mean reversion compared to the wstETH/WETH pair. Despite some variability, the order of magnitude of the coefficients remains consistent across months. Consequently, in the subsequent analysis, we use the rounded figures from Tables \ref{table:used-params-stable} and \ref{table:used-params-LST}, which align with our estimations.\\

\begin{table}[h!]
\renewcommand{\arraystretch}{1.2}
\begin{center}
\begin{tabular}{c|c||c|c|c}
\multicolumn{2}{c||}{1st-level OU} & \multicolumn{3}{c}{2nd-level OU}\\
\hline
Mean reversion&Volatility&Mean reversion&Volatility&Target\\
\hline
$\kappa$ &$\sigma$ &$\eta$ &$\nu$ &$\bar{U}$\\
\hline
 $5\times10^{-2}$ & $5\times10^{-4}$ & $3\times10^{-2}$ & $5\times10^{-4}$ & $1.00$ \\
\end{tabular}
\captionsetup{width=0.95\textwidth}
\caption{Parameters used for USDC/USDT. Units are day$^{-1}$ for the mean reversion parameters $\kappa$ and $\eta$, USDT $\cdot$ day$^{-\frac 12}$ for the volatility parameters $\sigma$ and $\nu$, and USDT for $\bar U$.}
\label{table:used-params-stable}
\end{center}
\end{table}

\begin{table}[h!]
\renewcommand{\arraystretch}{1.2}
\begin{center}
\begin{tabular}{c|c||c|c|c}
\multicolumn{2}{c||}{1st-level OU} & \multicolumn{3}{c}{2nd-level OU}\\
\hline
Mean reversion&Volatility&Mean reversion&Volatility&Target\\
\hline
$\kappa$ &$\sigma$ &$\eta$ &$\nu$ &$\bar{U}$\\
\hline
 $6$ & $6\times10^{-3}$ & $3$ & $4\times10^{-3}$ & $1.15$ \\
\end{tabular}
\captionsetup{width=0.95\textwidth}
\caption{Parameters used for wstETH/WETH. Units are day$^{-1}$ for the mean reversion parameters $\kappa$ and $\eta$, WETH $\cdot$ day$^{-\frac 12}$ for the volatility parameters $\sigma$ and $\nu$, and WETH for $\bar U$.}
\label{table:used-params-LST}
\end{center}
\end{table}

\subsection{Filtering: Estimating the Unobservable Price Target}

The multi-level nested OU process introduces the unobservable process \( (U_t)_t \), which acts as the dynamic mean-reversion target in price evolution. Since \( (U_t)_t \) cannot be directly observed, we employ the filtering technique described in Section \ref{sec:filter-procedure} to estimate its value.\\

Figure \ref{fig:filter-stable} displays the results of this filtering procedure for the USDC/USDT pair, showing the evolution of the estimated process \( (\widehat{U}_t)_t \) over the entire sample period, using the parameters listed in Table \ref{table:used-params-stable}. For the wstETH/WETH pair,\footnote{From this point forward, we work exclusively with discounted prices.} Figure \ref{fig:filter-LST} shows the corresponding filtered process \( (\widehat{U}_t)_t \), computed with the parameters given in Table \ref{table:used-params-LST}. To offer a more detailed perspective on the outcomes of the filtering process for this pair, Figure \ref{fig:filter-LST-zoom} zooms in on the first week of the sample period. Collectively, these figures highlight the effectiveness of the filtering technique in capturing the unobservable mean-reversion target across various pairs of pegged crypto assets.

\begin{figure}[!h]
    \centering
    \includegraphics[width=0.78\textwidth]{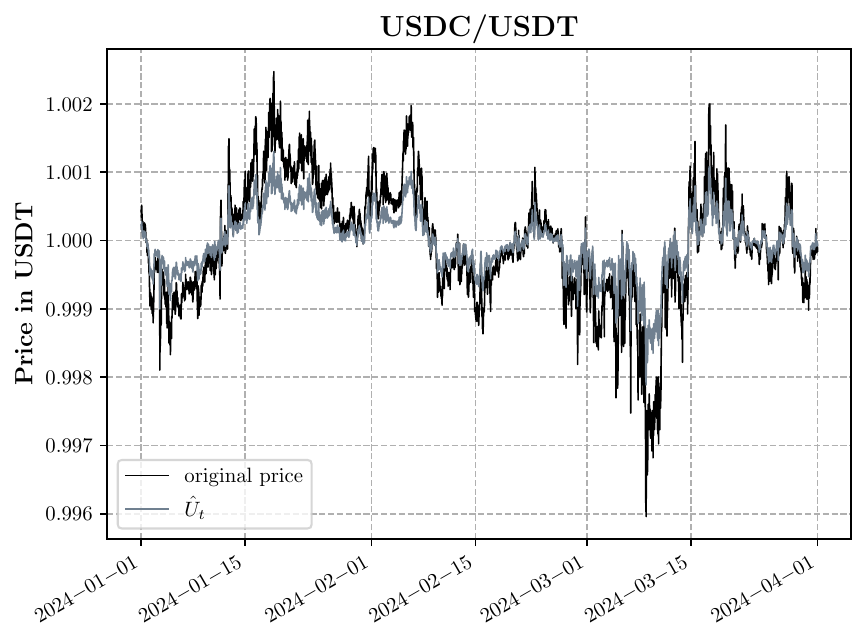}
    \caption{Results of the filtering procedure for USDC/USDT from January 1, 2024, to March 31, 2024. The dark line represents the original price, while the light line corresponds to the filtered series \( (\hat{U}_t)_t \), obtained using the parameters from Table \ref{table:used-params-stable}.}
    \label{fig:filter-stable}
\end{figure}

\begin{figure}[!h]
    \centering
    \includegraphics[width=0.78\textwidth]{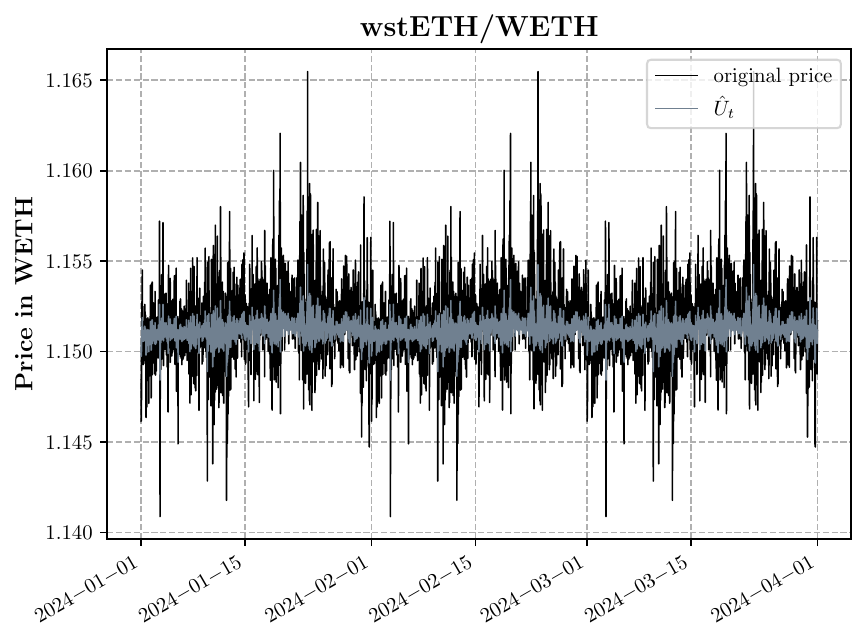}
    \caption{Results of the filtering procedure for wstETH/WETH from January 1, 2024, to March 31, 2024. The dark line represents the original price, while the light line corresponds to the filtered series \( (\hat{U}_t)_t \), obtained using the parameters from Table \ref{table:used-params-LST}.}
    \label{fig:filter-LST}
\end{figure}

\begin{figure}[!h]
    \centering
    \includegraphics[width=0.77\textwidth]{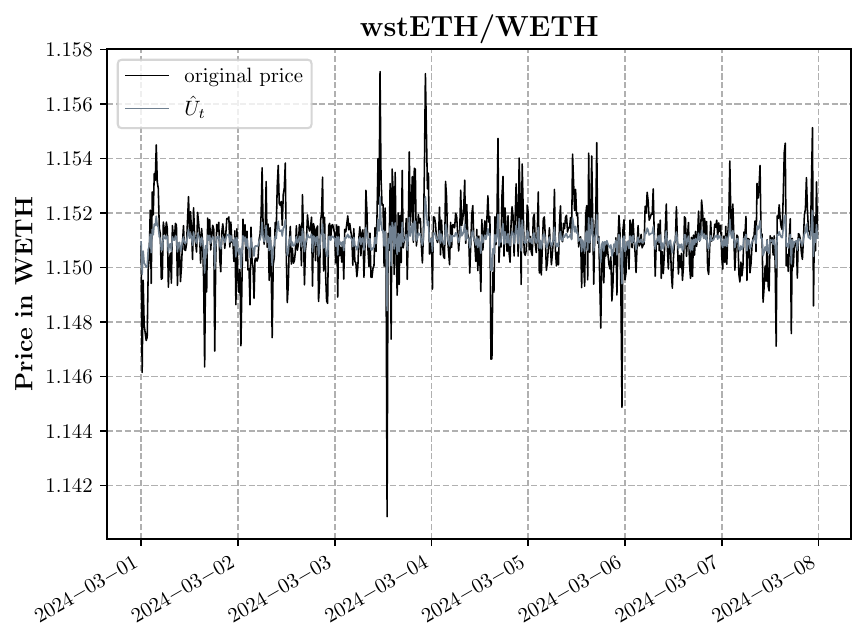}
    \caption{Results of the filtering procedure for wstETH/WETH from January 1, 2024, to January 7, 2024. The dark line represents the original price, while the light line corresponds to the filtered series \( (\hat{U}_t)_t \), obtained using the parameters from Table \ref{table:used-params-LST}.}
    \label{fig:filter-LST-zoom}
\end{figure}

\begin{figure}[!h]
    \centering
    \includegraphics[width=0.75\textwidth]{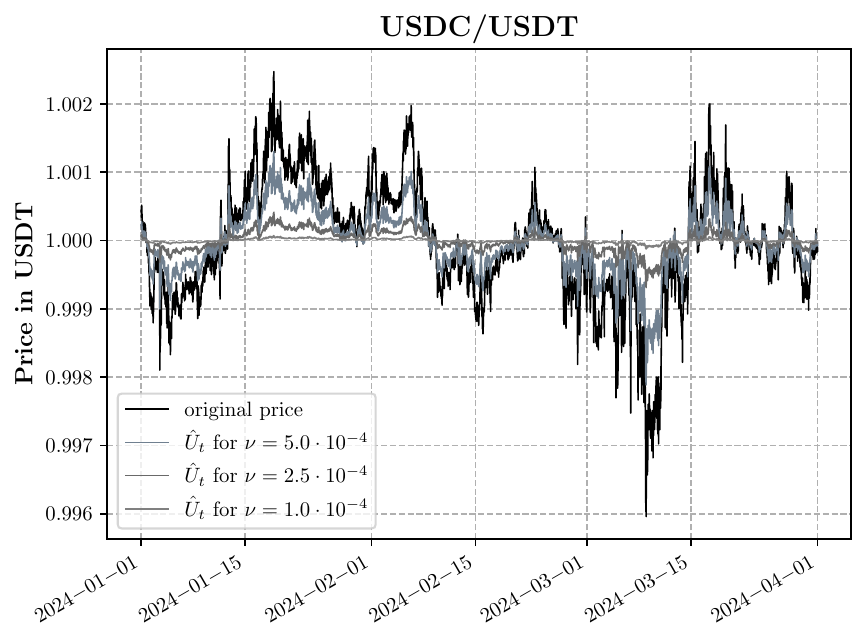}
    \caption{Results of the filtering procedure for USDC/USDT from January 1, 2024, to March 31, 2024 for several values of $\nu$. The dark line represents the original price, while the light lines correspond to the filtered series \( (\hat{U}_t)_t \), obtained using several values for $\nu$.}
    \label{fig:filter-many-nus1}
\end{figure}

In the preceding analysis, we applied the filtering procedure of Section \ref{sec:filter-procedure} using the parameters from Tables~\ref{table:used-params-stable} and \ref{table:used-params-LST}, which were determined through the statistical estimation procedure described in Section \ref{section:estim-nestedOU}. An interesting observation is that some of the parameters driving the underlying process \( (U_t)_t \) can be fixed rather than estimated. While this might initially seem counterintuitive, it can be viewed as a way to adjust the level of confidence in the target value of the actual price process.\\

To illustrate this, we ran the filtering algorithm with various values of \( \nu \). Figure \ref{fig:filter-many-nus1} presents the results for the USDC/USDT pair.\footnote{For readability, we focus solely on the stablecoin case here.} The results demonstrate that when \( \nu \) is large, the underlying process closely tracks the actual price, indicating low confidence in the peg (i.e., the long-term target). Conversely, when \( \nu \) is small, the hidden process remains near the long-term target, suggesting high confidence that the price process will revert in the short term towards \( \bar{U} \), as in a standard OU process.

\subsection{Numerical Simulation}

We now turn to the performance of the price-aware AMM strategy proposed in Section \ref{sec:optimal-model}: NOU AMM. Our analysis begins with simulated price data following a multi-level nested OU process and then extends to real price data.\\

To test the performance on simulated price data, we use the parameters from Table \ref{table:used-params-stable}, which correspond to the USDC/USDT pair. For the trade flow and liquidity parameters, we assume a fixed trade size of $100,000$~USDT and the following values:
\[
\lambda^{0,1} = \lambda^{1,0} = 250\ \text{day}^{-1}, \quad a^{0,1} = a^{1,0} = 0, \quad b^{0,1} = b^{1,0} = 10,000\ \text{USDT}^{-1}.
\]

These values imply an average of $125$ trades per day on each side when the proposed exchange rate matches the reference rate, $183$ trades per day when it offers a more favorable rate by $1$ basis point, and $67$ trades per day when it offers a less favorable rate by $1$ basis point.\\

We conducted simulations by running 300 paths of the multi-level nested OU process over one day and simulating the performance of the AMM strategy for a range of risk aversion parameters \( \gamma \) spanning the full spectrum of relevant values. The strategy was tested under a random trade flow characterized by the above liquidity parameters. For each \( \gamma \), we recorded the average excess P\&L and the standard deviation across the 300 simulations.  The simulation results are presented in Figure \ref{fig:NOUeff}. As expected, a highly risk-averse agent achieves negligible excess P\&L with minimal risk, while reducing the risk aversion parameter results in higher excess P\&L at the cost of increased standard deviation.\\

Unsurprisingly, the NOU AMM strategy, which is based on the multi-level nested OU dynamics used to simulate prices, outperforms the GBM AMM strategy proposed in \cite{bergault2024automated}, which assumes geometric Brownian motion dynamics.\footnote{The value of the volatility parameter does not matter for the GBM AMM because our simulations cover a large range of values of the risk aversion parameter $\gamma$.}\\

To test whether the proposed AMM strategy is effective in practice we then applied it to real price data. To this end, we used the same simulation setup as described above, replacing the simulated price paths with the actual USDC/USDT time series starting from 300 random times. Figure \ref{fig:eff-stable} compares the performance of the {NOU AMM} and GBM AMM strategies. The results closely mirror those obtained with simulated data, indicating that multi-level nested OU processes are particularly well-suited for modeling the risk and identifying trading opportunities associated with the USDC/USDT pair.\\

It is noteworthy that for both simulated and real USDC/USDT prices, the performance of the GBM AMM is not as good as that of the NOU AMM, yet remains respectable despite relying on a price dynamics that is misaligned with the actual behavior of USDC/USDT. This can be attributed to the time scale: given the liquidity parameters of the USDC/USDT pair, inventory round trips can occur over short intervals, during which the price dynamics can be reasonably approximated by a geometric Brownian motion.\\

\begin{figure}[!h]
    \centering
    \includegraphics[width=0.74\textwidth]{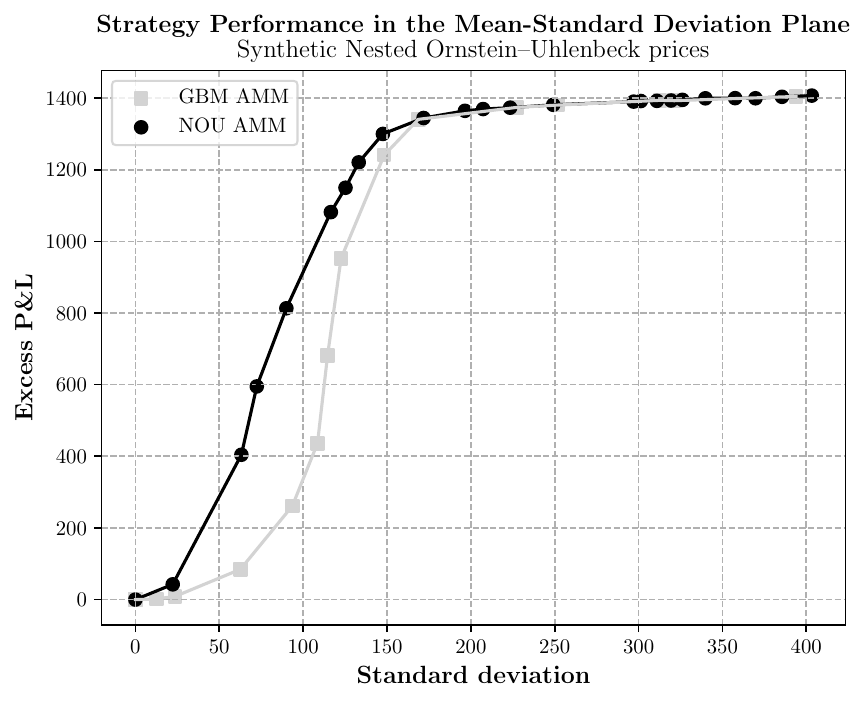}
    \caption{Strategy performance on simulated prices using a multi-level nested OU process with the parameters of Table \ref{table:used-params-stable} for different risk aversion parameters.}
    \label{fig:NOUeff}
\end{figure}

\begin{figure}[!h]
    \centering
    \includegraphics[width=0.74\textwidth]{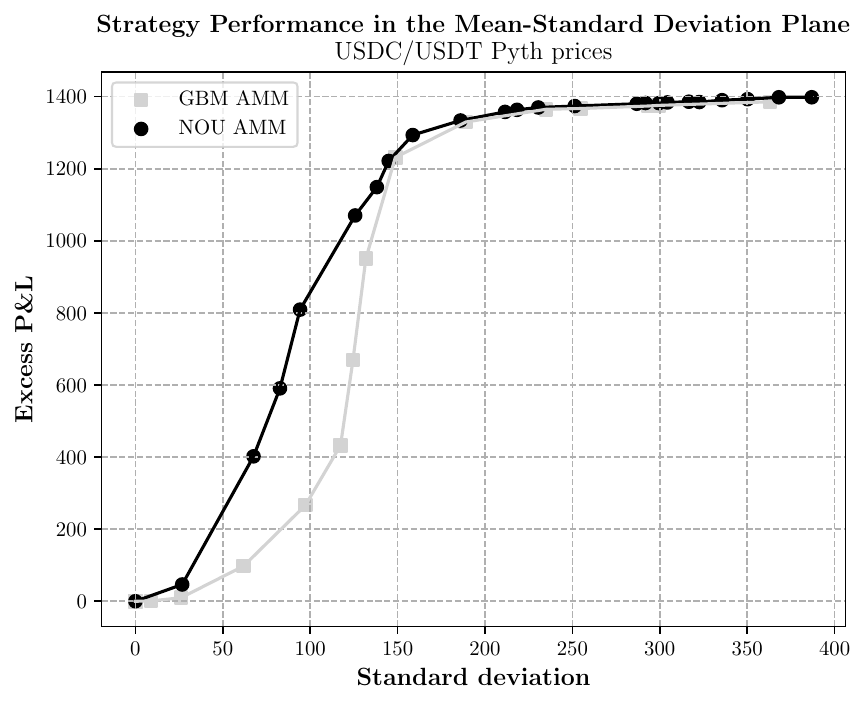}
    \caption{Strategy performance on real USDC/USDT data for different risk aversion parameters.}
    \label{fig:eff-stable}
\end{figure}

For the wstETH/WETH pair, we carried out the same simulations, directly on real price data, by using the NOU AMM strategy for the price parameters of Table \ref{table:used-params-LST} and trade flow and liquidity parameters given by a fixed trade size of $40$ WETH and the following values:
\[
\lambda^{0,1} = \lambda^{1,0} = 250\ \text{day}^{-1}, \quad a^{0,1} = a^{1,0} = 0, \quad b^{0,1} = b^{1,0} = 10,000\ \text{WETH}^{-1}.
\]

The simulation results are presented in Figure \ref{fig:eff-lst}. As before, a highly risk-averse agent achieves negligible excess P\&L with minimal risk, while reducing the risk aversion parameter leads to higher excess P\&L at the cost of increased standard deviation. The decrease in the performance of the NOU AMM on the right side of the plot for small values of \( \gamma \) is merely a result of random liquidity-taking and is not statistically significant.\\

Unlike the USDC/USDT case, the performance of the {GBM AMM} is significantly inferior to that of the {NOU AMM} for the wstETH/WETH pair. The fast mean-reverting behavior of wstETH/WETH renders geometric Brownian motion ill-suited to model the price process, even on the short time scale of inventory round trips. The results clearly show that the {NOU AMM} leverages the multi-level nested OU dynamics to achieve levels of average excess P\&L that outperform those of the {GBM AMM} by far.

\begin{figure}[h]
    \centering
    \includegraphics[width=0.75\textwidth]{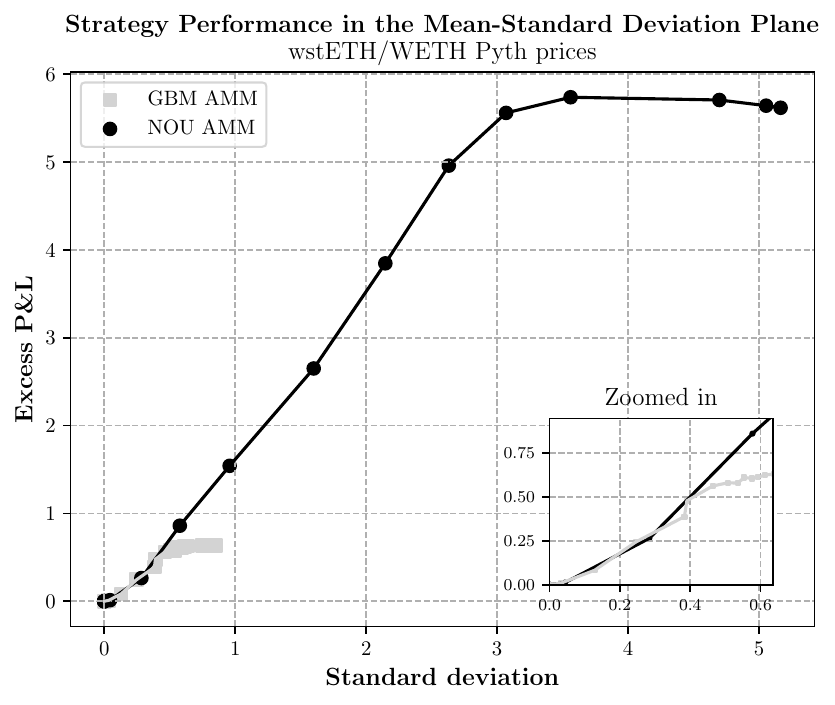}
    \caption{Strategy performance on real wstETH/WETH prices for different risk aversion parameters.}
    \label{fig:eff-lst}
\end{figure}

\section*{Conclusion}

As the cryptocurrency and decentralized finance ecosystems continue to expand, stablecoins and liquid staking tokens are attracting significant trading volumes. Developing effective liquidity provision algorithms tailored to these specific assets is essential, as they differ markedly from traditional cryptos in their mean-reverting price dynamics.\\

In this paper, we introduced a framework that employs multi-level nested Ornstein-Uhlenbeck processes to capture the exchange rate dynamics of intrinsically linked assets, such as stablecoins pegged to the same fiat currency and liquid staking tokens alongside their native tokens. The automated market maker (AMM) strategy we developed within this framework leverages mean-reverting dynamics, offering significant improvements in profitability and risk management compared to existing approaches. Our findings underscore the potential of this model to enhance liquidity in markets for pegged assets, paving the way for more robust trading mechanisms as crypto markets continue to evolve.

\section*{Statement and Acknowledgment}

This research was supported by the Research Program ``Decentralized Finance and Automated Market Makers,'' a program under the aegis of Institut Europlace de Finance, in partnership with Apesu / Swaap Labs.\\

Olivier Guéant would like to express his gratitude to Arthur Charpentier, Jean-David Fermanian, Christian Francq, Christian Gouriéroux, and Jean-Michel Zakoïan for their insightful discussions on the generalizations of OU processes. Philippe Bergault and Olivier Guéant also wish to thank Alexander (Sasha) Barzykin for their stimulating discussions on stochastic filtering.\\

\nocite{*}
\bibliographystyle{plain}

\newpage

\appendix

\section*{Appendix on Stochastic Filtering}

In this appendix, we provide essential results of stochastic filtering. We consider a filtered probability space \((\Omega, \mathbb{F} = (\mathcal{F}_t)_t, \mathbb{P})\) that satisfies the usual assumptions. All random variables used hereafter are assumed to be defined on this probability space and the initial stochastic processes (before filtering) adapted to the filtration.\\

We consider two stochastic processes $(Z_t)_t$ and $(\zeta_t)_t$ with values in $\mathbb R^k$ and $\mathbb R^d$ respectively, solutions of the following stochastic differential equations:
\begin{align*}
    dZ_t &= \Gamma \begin{pmatrix} Z_t\\ \zeta_t \end{pmatrix} dt + {\Sigma_Z^{\frac 12}} dW^Z_t\\
    d\zeta_t &= (\Theta\zeta_t+\upsilon) dt + {\Sigma_\zeta^{\frac 12}} dW^\zeta_t,
\end{align*}
where $\begin{pmatrix}W^Z \\ W^\zeta \end{pmatrix}$ is a $k+d$-dimensional Brownian motion with covariance structure given by
$$\rho = \left(\begin{array}{c|c}
I_k & \tilde \rho \\
\hline
\tilde \rho^\intercal & I_d\\
\end{array}\right),$$
where $\tilde \rho \in \mathcal M_{k,d}(\mathbb R)$.\\

For the dynamics of $(Z_t)_t$, we used two matrices $\Gamma\in \mathcal M_{k, k+d}(\mathbb R)$ and ${\Sigma_Z} \in \mathcal S^{++}_{k}(\mathbb R)$. For the dynamics of $(\zeta_t)_t$, we used the matrices $\Theta\in \mathcal M_{d, d}(\mathbb R)$ and ${\Sigma_\zeta} \in \mathcal S^+_{d}(\mathbb R)$, and the vector $\upsilon\in \mathbb R^d$.

\begin{rem}
In what follows, we assume that $Z_0$ and $\zeta_0$ are respectively constant and Gaussian. Subsequently the subspace of $L^2(\Omega)$ spanned by $\{Z_t, \zeta_t | t \ge 0\}$ is a Gaussian space.
\end{rem}

We assume that we only observe the process $(Z_t)_t$, and we want to estimate the value of the process $(\zeta_t)_t$ at each point in time. Denoting by $\mathbb F^Z = (\mathcal F^Z_t)_t$ the natural filtration associated with $(Z_t)_t$, our process of interest is $(\widehat \zeta_t)_t$ defined by
$$\widehat \zeta_t = \mathbb E \left[\zeta_t | \mathcal F^Z_t \right].$$

We start with a first lemma on martingales:

\begin{lemma}
\label{lemma_martingale}
Let $(\Psi_t)_t$ and $(\psi_t)_t$ be two $\mathbb F$-adapted processes. We further assume that for all $T >0$,
$\mathbb E \left[\displaystyle\int_0^T \| \psi_t\| dt \right]<+\infty$.\\

If the process $\left({\Psi}_t  - \displaystyle\int_0^t \psi_s ds\right)_t$ is an $\mathbb F$-martingale, then the process 
$\left(\mathbb{E}[{\Psi}_t|\mathcal F_t^Z] - \displaystyle\int_0^t  \mathbb{E}[\psi_s|\mathcal F_s^Z]  ds \right)_t  $ is well defined and is an $\mathbb F^Z$-martingale.
\end{lemma}

\begin{proof}
If $\left({\Psi}_t  - \displaystyle\int_0^t \psi_s ds\right)_t$ is an $\mathbb F$-martingale, then the process 
$\left(\mathbb{E}[{\Psi}_t|\mathcal F_t^Z] - \displaystyle\int_0^t  \mathbb{E}[\psi_s|\mathcal F_t^Z]  ds \right)_t  $ is well defined and is an $\mathbb F^Z$-martingale.\\

We therefore only need to prove that the process 
$\left(\displaystyle\int_0^t  \left(\mathbb{E}[\psi_s|\mathcal F_t^Z] - \mathbb{E}[\psi_s|\mathcal F_s^Z]\right) ds \right)_t  $
is an $\mathbb F^Z$-martingale.\\

But, for all $t, h>0$, we have indeed
\begin{eqnarray*}
&&\mathbb{E}\left[\left.\int_0^{t+h}  \left(\mathbb{E}[\psi_s|\mathcal F_{t+h}^Z] - \mathbb{E}[\psi_s|\mathcal F_s^Z]\right) ds\right|\mathcal F_t^Z\right]\\
&=& \mathbb{E}\left[\left.\int_0^{t}  \left(\mathbb{E}[\psi_s|\mathcal F_{t+h}^Z] - \mathbb{E}[\psi_s|\mathcal F_s^Z]\right) ds\right|\mathcal F_t^Z\right] + \mathbb{E}\left[\left.\int_t^{t+h}  \left(\mathbb{E}[\psi_s|\mathcal F_{t+h}^Z] - \mathbb{E}[\psi_s|\mathcal F_s^Z]\right) ds\right|\mathcal F_t^Z\right] \\
&=& \int_0^{t}  \left(\mathbb{E}[\psi_s|\mathcal F_{t}^Z] - \mathbb{E}[\psi_s|\mathcal F_s^Z]\right) ds + \int_t^{t+h}  \left(\mathbb{E}[\psi_s|\mathcal F_{t}^Z] - \mathbb{E}[\psi_s|\mathcal F_t^Z]\right) ds \\
&=& \int_0^{t}  \left(\mathbb{E}[\psi_s|\mathcal F_{t}^Z] - \mathbb{E}[\psi_s|\mathcal F_s^Z]\right) ds,\\
\end{eqnarray*}
hence the result.\\
\end{proof}

We can then state a proposition related to the innovation process:

\begin{prop}
Let us denote by $(\widehat W^Z_t)_t$ the process defined by
$$\widehat W^Z_t = W^Z_t + \int_0^t {\Sigma_Z^{-\frac 12}}\Gamma \begin{pmatrix}
    0\\
     \zeta_s - \widehat \zeta_s
\end{pmatrix}ds.$$
Then, $(\widehat W^Z_t)_t$ is an $\mathbb F^Z$-Brownian motion. Moreover, $\mathbb F^Z$ is the natural filtration associated with $(\widehat W^Z_t)_t$.
\end{prop}
\begin{proof}
The process $\left(Z_t - \displaystyle \int_0^t \Gamma \begin{pmatrix}
    Z_s\\
     \zeta_s
\end{pmatrix}ds \right)_t$ is an $\mathbb F$-martingale. Using the result of Lemma \ref{lemma_martingale}, the process $\left(Z_t - \displaystyle \int_0^t \Gamma \begin{pmatrix}
    Z_s\\
     \widehat \zeta_s
\end{pmatrix}ds \right)_t$
is an $\mathbb F^Z$-martingale.\\

Since, $\widehat W^Z_t = {\Sigma_Z^{-\frac 12}}\left(Z_t - \displaystyle \int_0^t \Gamma \begin{pmatrix} Z_s\\ \widehat \zeta_s \end{pmatrix} ds\right)$, the process $(\widehat W^Z_t)_t$ is an $\mathbb F^Z$-martingale. Lévy's characterization then enables to conclude that $(\widehat W^Z_t)_t$ is an $\mathbb F^Z$-Brownian motion.\\

The last statement follows from the stochastic differential equation
$$dZ_t = \Gamma \begin{pmatrix} Z_t\\ \widehat \zeta_t \end{pmatrix} dt + {\Sigma_Z^{\frac 12}} d\widehat W^Z_t.$$
\end{proof}

Using the $\mathbb F^Z$-Brownian motion $(\widehat W^Z_t)_t$, we can state the stochastic differential equation corresponding to the process $(\widehat \zeta_t)_t$.\\

\begin{prop}
\label{zetachapo}
$(\widehat \zeta_t)_t$ solves the stochastic differential equation
$$ d\widehat \zeta_t = (\Theta\widehat \zeta_t+\upsilon) dt  +  \psi_t d\widehat W^Z_t,$$
where $$\psi_t = \begin{pmatrix} 0\\ \mathbb{V}( \zeta_t |\mathcal F_t^Z)\end{pmatrix}^\intercal \Gamma^\intercal {\Sigma_Z^{-\frac 12}}+ {\Sigma_\zeta^{\frac 12}}\tilde \rho^\intercal,$$ with $\mathbb{V}( \cdot |\mathcal F_t^Z)$ standing for the conditional variance-covariance matrix.\\
\end{prop}

\begin{proof}
Using Lemma \ref{lemma_martingale}, the process 
$$\left(\widehat \zeta_t - \displaystyle \int_0^t (\Theta\widehat \zeta_s  + \upsilon)ds \right)_t$$
is a square-integrable $\mathbb F^Z$-martingale and, therefore, from the martingale representation theorem, there exists an $\mathbb F^Z$-predictable process $(\psi_t)_t$ with values in $\mathcal M_{d,k}(\mathbb R)$ such that for all $T>0$, 
$$\mathbb E \left[\int_0^T \|\psi_t\|^2 dt \right] <+\infty$$
and the process $(\widehat \zeta_t)_t$ is solution of the stochastic differential equation
\begin{equation*}
    d\widehat \zeta_t = (\Theta\widehat \zeta_t  + \upsilon) dt +  \psi_t d \widehat W^Z_t.
\end{equation*}

Let $M\in \mathcal M_{k,d}(\mathbb R)$. We have
\begin{align*}
    d(Z_t^\intercal M\zeta_t) &= dZ_t^\intercal M\zeta_t + Z_t^\intercal Md\zeta_t + \text{Tr}\left(M{\Sigma_\zeta^{\frac 12}}\tilde \rho^\intercal {\Sigma_Z^{\frac 12}} \right) dt\\
    &= \left( \begin{pmatrix} Z_t\\ \zeta_t \end{pmatrix}^\intercal \Gamma^\intercal M\zeta_t + Z_t^\intercal M (\Theta\zeta_t+\upsilon) + \text{Tr}\left(M{\Sigma_\zeta^{\frac 12}}\tilde \rho^\intercal {\Sigma_Z^{\frac 12}}\right)\right) dt + Z_t^\intercal M{\Sigma_\zeta^{\frac 12}} dW_t^\zeta + \zeta^\intercal_tM^\intercal{\Sigma_Z^{\frac 12}} {dW_t^Z},
\end{align*}
so the process $$\left(Z_t^\intercal M\zeta_t - \displaystyle \int_0^t \left(  \begin{pmatrix} Z_s\\ \zeta_s \end{pmatrix}^\intercal \Gamma^\intercal M\zeta_s + Z_s^\intercal M (\Theta\zeta_s+\upsilon) + \text{Tr}\left(M{\Sigma_\zeta^{\frac 12}}\tilde \rho^\intercal {\Sigma_Z^{\frac 12}}\right)\right) ds  \right)_t$$ is an $\mathbb F$-martingale and, by Lemma \ref{lemma_martingale}, the process
$$\left(Z_t^\intercal M\widehat\zeta_t - \displaystyle \int_0^t \left(   \mathbb E\left[\left.\begin{pmatrix} Z_s\\ \zeta_s \end{pmatrix}^\intercal \Gamma^\intercal M\zeta_s\right| \mathcal F_s^Z \right] + Z_s^\intercal M (\Theta\widehat\zeta_s+\upsilon) + \text{Tr}\left(M{\Sigma_\zeta^{\frac 12}}\tilde \rho^\intercal {\Sigma_Z^{\frac 12}}\right)\right) ds  \right)_t$$ is an $\mathbb F^Z$-martingale.\\

On the other hand,
\begin{align*}
    d(Z_t^\intercal M\widehat \zeta_t) &= dZ_t^\intercal M\widehat \zeta_t +  Z_t^\intercal Md\widehat \zeta_t + \text{Tr}\left(M \psi_t {\Sigma_Z^{\frac 12}}  \right) dt\\
    &=  \left(\begin{pmatrix} Z_t\\ \widehat \zeta_t \end{pmatrix}^\intercal \Gamma^\intercal M\widehat \zeta_t  + Z_t^\intercal M (\Theta\widehat \zeta_t+\upsilon) +  \text{Tr}\left(M\psi_t{\Sigma_Z^{\frac 12}} \right) \right) dt + \left(Z_t^\intercal M\psi_t + \widehat \zeta_t^\intercal M^\intercal {\Sigma_Z^{\frac 12}} \right) d\widehat W_t^Z 
\end{align*}
so the process
$$\left(Z_t^\intercal M\widehat\zeta_t - \displaystyle \int_0^t \left(   \begin{pmatrix} Z_s\\ \widehat \zeta_s \end{pmatrix}^\intercal \Gamma^\intercal M\widehat \zeta_s + Z_s^\intercal M (\Theta\widehat\zeta_s+\upsilon) + \text{Tr}\left(M\psi_s{\Sigma_Z^{\frac 12}} \right)\right) ds  \right)_t$$ is also an $\mathbb F^Z$-martingale.\\

Therefore, for all $M\in \mathcal M_{k,d}(\mathbb R)$, we have
\begin{eqnarray*}
\text{Tr}\left(M\psi_t{\Sigma_Z^{\frac 12}} \right) &=& \mathbb E\left[\left.\begin{pmatrix} Z_t\\ \zeta_t \end{pmatrix}^\intercal \Gamma^\intercal M\zeta_t\right| \mathcal F_t^Z \right] - \begin{pmatrix} Z_t\\ \widehat \zeta_t \end{pmatrix}^\intercal \Gamma^\intercal M\widehat \zeta_t + \text{Tr}\left(M{\Sigma_\zeta^{\frac 12}}\tilde \rho^\intercal  {\Sigma_Z^{\frac 12}}\right)\\
&=&\text{Tr}\left(\mathbb E\left[\left.\begin{pmatrix} Z_t\\ \zeta_t \end{pmatrix}^\intercal \Gamma^\intercal M\zeta_t\right| \mathcal F_t^Z \right]\right) - \text{Tr} \left(\begin{pmatrix} Z_t\\ \widehat \zeta_t \end{pmatrix}^\intercal \Gamma^\intercal M\widehat \zeta_t\right) + \text{Tr}\left(M{\Sigma_\zeta^{\frac 12}}\tilde \rho^\intercal  {\Sigma_Z^{\frac 12}} \right)\\
&=& \text{Tr}\left(M \left( \mathbb E\left[\left.\zeta_t\begin{pmatrix} Z_t\\ \zeta_t \end{pmatrix}^\intercal -  \widehat \zeta_t\begin{pmatrix} Z_t\\ \widehat \zeta_t \end{pmatrix}^\intercal \right|\mathcal F_t^Z \right] \Gamma^\intercal + {\Sigma_\zeta^{\frac 12}}\tilde \rho^\intercal  {\Sigma_Z^{\frac 12}}\right)  \right)\\
&=& \text{Tr}\left(M \left( \begin{pmatrix} 0\\ \mathbb V( \zeta_t |\mathcal F_t^Z)\end{pmatrix}^\intercal \Gamma^\intercal + {\Sigma_\zeta^{\frac 12}}\tilde \rho^\intercal  {\Sigma_Z^{\frac 12}}\right)  \right).\\
\end{eqnarray*}

Therefore, we must have 
$$\psi_t = \begin{pmatrix} 0\\ \mathbb V( \zeta_t |\mathcal F_t^Z)\end{pmatrix}^\intercal \Gamma^\intercal {\Sigma_Z^{-\frac 12}}+ {\Sigma_\zeta^{\frac 12}}\tilde \rho^\intercal.$$
\end{proof}

The above proposition involves the conditional variance-covariance term $\mathbb V( \zeta_t |\mathcal F_t^Z)$ that needs to be computed. This is the purpose of the following proposition.

\begin{prop}
Let $V_t = \mathbb V( \zeta_t |\mathcal F_t^Z)$.\\

$(V_t)_t$ is a deterministic process satisfying the following matrix Riccati diffential equation:
$$\frac{dV_t}{dt} = \Theta V_t + V_t \Theta^\intercal + \Sigma_\zeta - \left(\begin{pmatrix} 0\\ V_t\end{pmatrix}^\intercal \Gamma^\intercal {\Sigma_Z^{-\frac 12}}+ {\Sigma_\zeta^{\frac 12}}\tilde \rho^\intercal\right)\left(\begin{pmatrix} 0\\ V_t\end{pmatrix}^\intercal \Gamma^\intercal {\Sigma_Z^{-\frac 12}}+ {\Sigma_\zeta^{\frac 12}}\tilde \rho^\intercal\right)^\intercal.$$   
\end{prop}

\begin{proof}
Applying Ito's formula, we get
\begin{eqnarray*}
d(\zeta_t\zeta_t^\intercal) &=& (\Theta\zeta_t\zeta_t^\intercal+\zeta_t\zeta_t^\intercal\Theta^\intercal+\upsilon\zeta_t^\intercal + \zeta_t \upsilon^\intercal + \Sigma_\zeta) dt + {\Sigma_\zeta^{\frac 12}} dW^\zeta_t\zeta_t^\intercal + \zeta_t{dW^\zeta_t}^\intercal{\Sigma_\zeta^{\frac 12}}.
\end{eqnarray*}
Therefore, for each couple $(i,j) \in \{1, \ldots d\}^2$,\footnote{$(e^1, \ldots, e^d)$ denotes the canonical basis of $\mathbb R^d$.} we have
\begin{eqnarray*}
d(\zeta^i_t\zeta^j_t) &=& ({e^i}^\intercal\Theta\zeta_t\zeta^j_t+\zeta^i_t\zeta_t^\intercal\Theta^\intercal e^j+\upsilon^i\zeta^j_t + \zeta^i_t \upsilon^j + \Sigma^{i,j}_\zeta) dt + \left(\Sigma_\zeta^{\frac 12}\left(\zeta_t^i e^j+ \zeta_t^j e^i \right)\right)^\intercal dW^\zeta_t.
\end{eqnarray*}
The process $$\left(\zeta^i_t\zeta^j_t - \int_0^t ({e^i}^\intercal\Theta\zeta_s\zeta^j_s+\zeta^i_s\zeta_s^\intercal\Theta^\intercal e^j+\upsilon^i\zeta^j_s + \zeta^i_s \upsilon^j + \Sigma^{i,j}_\zeta)\right)_t$$ is thus an $\mathbb F$-martingale, and therefore, using Lemma \ref{lemma_martingale}, the process $$\left(\mathbb E\left[\left.\zeta^i_t\zeta^j_t\right|\mathcal F_t^Z \right] - \int_0^t \left({e^i}^\intercal\Theta\mathbb E\left[\left.\zeta_s\zeta^j_s\right|\mathcal F_s^Z \right]+\mathbb E\left[\left.\zeta^i_s\zeta_s^\intercal\right|\mathcal F_s^Z \right]\Theta^\intercal e^j+\upsilon^i\widehat \zeta^j_s + \widehat \zeta^i_s \upsilon^j + \Sigma^{i,j}_\zeta\right) ds \right)_t$$
is a square-integrable $\mathbb F^Z$-martingale. From the martingale representation theorem, there exists for each couple $ (i,j) \in \{1, \ldots d\}^2$ an $\mathbb F^Z$-predictable process $(\phi^{i,j}_t)_t$ with value in $\mathbb R^k$ such that for all $T>0$, 
$$\mathbb E \left[\int_0^T \|\phi^{i,j}_t\|^2 dt \right] <+\infty$$
and
\begin{eqnarray}
d\mathbb E\left[\left.\zeta^i_t\zeta^j_t\right|\mathcal F_t^Z \right]\! &=&\! \left({e^i}^\intercal\Theta\mathbb E\left[\left.\zeta_t\zeta^j_t\right|\mathcal F_t^Z \right]+\mathbb E\left[\left.\zeta^i_t\zeta_t^\intercal\right|\mathcal F_t^Z \right]\Theta^\intercal e^j+\upsilon^i\widehat \zeta^j_t + \widehat \zeta^i_t \upsilon^j + \Sigma^{i,j}_\zeta\right) dt + {\phi^{i,j}_t}^\intercal d\widehat W^Z_t.\label{eofsquare}
\end{eqnarray}

Now, using Proposition \ref{zetachapo} and Ito's formula, we have
\begin{eqnarray*}
d(\widehat\zeta_t\widehat\zeta_t^\intercal) &=& \left(\Theta\widehat\zeta_t\widehat\zeta_t^\intercal+\widehat\zeta_t\widehat\zeta_t^\intercal\Theta^\intercal+\upsilon\widehat\zeta_t^\intercal + \widehat\zeta_t \upsilon^\intercal + \psi_t\psi_t^\intercal\right) dt + \psi_t d\widehat W^Z_t\widehat\zeta_t^\intercal + \widehat\zeta_t{d\widehat W^Z_t}^\intercal\psi_t^\intercal,
\end{eqnarray*}
so, for each couple $ (i,j) \in \{1, \ldots d\}^2$, we get
\begin{eqnarray}
d(\widehat\zeta^i_t\widehat\zeta^j_t) &=& \left({e^i}^\intercal\Theta\widehat\zeta_t\widehat\zeta^j_t+\widehat\zeta^i_t\widehat\zeta_t^\intercal\Theta^\intercal e^j+\upsilon^i\widehat\zeta^j_t + \widehat\zeta^i_t \upsilon^j + {e^i}^\intercal\psi_t\psi_t^\intercal e^j\right) dt + \left(\psi^\intercal_t \left(\widehat\zeta_t^i e^j+ \widehat\zeta_t^j e^i \right)\right)^\intercal d\widehat W^Z_t. \label{esquared}
\end{eqnarray}

Subtracting Eq. \eqref{esquared} from Eq. \eqref{eofsquare}, we get $\forall (i,j) \in \{1, \ldots d\}^2$,
$$d\nu^{i,j}_t = \left( {e^i}^\intercal\Theta V_t e^j + {e^i}^\intercal V_t \Theta^\intercal e^j + \Sigma^{i,j}_\zeta - {e^i}^\intercal\psi_t\psi_t^\intercal e^j\right) dt + \left(\phi^{i,j}_t - \psi^\intercal_t \left(\widehat\zeta_t^i e^j+ \widehat\zeta_t^j e^i \right)\right)^\intercal d\widehat W^Z_t,$$
i.e.
$${e^i}^\intercal\left( dV_t - \left( \Theta V_t + V_t \Theta^\intercal + \Sigma_\zeta - \psi_t\psi_t^\intercal\right) dt\right)e^j = \left(\phi^{i,j}_t - \psi^\intercal_t \left(\widehat\zeta_t^i e^j+ \widehat\zeta_t^j e^i \right)\right)^\intercal d\widehat W^Z_t.$$

Given the definition of $\psi_t$, the statement of the proposition is equivalent to $$\forall (i,j) \in \{1, \ldots d\}^2, \phi^{i,j}_t = \psi^\intercal_t \left(\widehat\zeta_t^i e^j+ \widehat\zeta_t^j e^i \right).$$

To prove the latter, let us apply Ito's formula to get
\begin{eqnarray*}
d(\zeta^i_t\zeta^j_tZ_t) &=& ({e^i}^\intercal\Theta\zeta_t\zeta^j_t+\zeta^i_t\zeta_t^\intercal\Theta^\intercal e^j+\upsilon^i\zeta^j_t + \zeta^i_t \upsilon^j + \Sigma^{i,j}_\zeta) Z_t dt + Z_t\left(\Sigma_\zeta^{\frac 12}\left(\zeta_t^i e^j+ \zeta_t^j e^i \right)\right)^\intercal dW^\zeta_t\\
&& + \zeta^i_t\zeta^j_t \Gamma \begin{pmatrix} Z_t\\ \zeta_t \end{pmatrix} dt + \zeta^i_t\zeta^j_t {\Sigma_Z^{\frac 12}} dW^Z_t + {\Sigma_Z^{\frac 12}}\tilde \rho\Sigma_\zeta^{\frac 12}\left(\zeta_t^i e^j+ \zeta_t^j e^i \right) dt.
\end{eqnarray*}

We deduce that the process 
$$\left(\zeta^i_t\zeta^j_tZ_t - \int_0^t \left( \left({e^i}^\intercal\Theta\zeta_s\zeta^j_s+\zeta^i_s\zeta_s^\intercal\Theta^\intercal e^j+\upsilon^i\zeta^j_s + \zeta^i_s \upsilon^j + \Sigma^{i,j}_\zeta\right) Z_s + \zeta^i_s\zeta^j_s \Gamma \begin{pmatrix} Z_s\\ \zeta_s \end{pmatrix} + {\Sigma_Z^{\frac 12}}\tilde \rho\Sigma_\zeta^{\frac 12}\left(\zeta_s^i e^j+ \zeta_s^j e^i \right)  \right) ds\right)_t$$
is an $\mathbb F$-martingale, and therefore, using Lemma \ref{lemma_martingale}, the process
$$\Bigg(\mathbb E\left[\left.\zeta^i_t\zeta^j_t\right|\mathcal F_t^Z \right]Z_t - \int_0^t \Big( \Big({e^i}^\intercal\Theta\mathbb E\left[\left.\zeta_s\zeta^j_s\right|\mathcal F_s^Z \right]+\mathbb E\left[\left.\zeta^i_s\zeta_s^\intercal\right|\mathcal F_s^Z \right]\Theta^\intercal e^j+\upsilon^i\widehat\zeta^j_s + \widehat\zeta^i_s \upsilon^j + \Sigma^{i,j}_\zeta\Big) Z_s$$
$$ + \mathbb E\left[\left.\zeta^i_s\zeta^j_s \Gamma \begin{pmatrix} Z_s\\ \zeta_s \end{pmatrix}\right|\mathcal F_s^Z \right] + {\Sigma_Z^{\frac 12}}\tilde \rho\Sigma_\zeta^{\frac 12}\left(\widehat\zeta_s^i e^j+ \widehat\zeta_s^j e^i \right)  \Big) ds\Bigg)_t$$
is an $\mathbb F^Z$-martingale.\\

On the other hand, using Eq. \eqref{eofsquare} and Ito's formula, we get

\begin{eqnarray*}
d\left(\mathbb E\left[\left.\zeta^i_t\zeta^j_t\right|\mathcal F_t^Z \right]Z_t\right) &=& \left({e^i}^\intercal\Theta\mathbb E\left[\left.\zeta_t\zeta^j_t\right|\mathcal F_t^Z \right]+\mathbb E\left[\left.\zeta^i_t\zeta_t^\intercal\right|\mathcal F_t^Z \right]\Theta^\intercal e^j+\upsilon^i\widehat \zeta^j_t + \widehat \zeta^i_t \upsilon^j + \Sigma^{i,j}_\zeta\right) Z_t dt + Z_t {\phi^{i,j}_t}^\intercal d\widehat W^Z_t\\
&& + \mathbb E\left[\left.\zeta^i_t\zeta^j_t\right|\mathcal F_t^Z \right] \Gamma \begin{pmatrix} Z_t\\ \widehat\zeta_t \end{pmatrix} dt + E\left[\left.\zeta^i_t\zeta^j_t\right|\mathcal F_t^Z \right] {\Sigma_Z^{\frac 12}} d\widehat W^Z_t + {\Sigma_Z^{\frac 12}}\phi^{i,j}_t dt.
\end{eqnarray*}
Therefore, the process
$$\Bigg(\mathbb E\left[\left.\zeta^i_t\zeta^j_t\right|\mathcal F_t^Z \right]Z_t - \int_0^t \Big( \Big({e^i}^\intercal\Theta\mathbb E\left[\left.\zeta_s\zeta^j_s\right|\mathcal F_s^Z \right]+\mathbb E\left[\left.\zeta^i_s\zeta_s^\intercal\right|\mathcal F_s^Z \right]\Theta^\intercal e^j+\upsilon^i\widehat\zeta^j_s + \widehat\zeta^i_s \upsilon^j + \Sigma^{i,j}_\zeta\Big) Z_s$$
$$ + \mathbb E\left[\left.\zeta^i_s\zeta^j_s\right|\mathcal F_s^Z \right] \Gamma \begin{pmatrix} Z_s\\ \widehat\zeta_s \end{pmatrix} + {\Sigma_Z^{\frac 12}} \phi^{i,j}_s  \Big) ds\Bigg)_t$$
is also an $\mathbb F^Z$-martingale.\\

We deduce that
$$ \mathbb E\left[\left.\zeta^i_t\zeta^j_t\right|\mathcal F_t^Z \right] \Gamma \begin{pmatrix} Z_t\\ \widehat\zeta_t \end{pmatrix} + {\Sigma_Z^{\frac 12}} \phi^{i,j}_t = \mathbb E\left[\left.\zeta^i_t\zeta^j_t \Gamma \begin{pmatrix} Z_t\\ \zeta_t \end{pmatrix}\right|\mathcal F_t^Z \right] + {\Sigma_Z^{\frac 12}}\tilde \rho\Sigma_\zeta^{\frac 12}\left(\widehat\zeta_t^i e^j+ \widehat\zeta_t^j e^i \right),$$
i.e.
$$\phi^{i,j}_t = {\Sigma_Z^{-\frac 12}}\left(\Gamma \begin{pmatrix} 0\\  \mathbb E\left[\left.\zeta^i_t\zeta^j_t\zeta_t \right|\mathcal F_t^Z \right] - \mathbb E\left[\left.\zeta^i_t\zeta^j_t \right|\mathcal F_t^Z \right]\widehat\zeta_t  \end{pmatrix} + {\Sigma_Z^{\frac 12}}\tilde \rho\Sigma_\zeta^{\frac 12}\left(\widehat\zeta_t^i e^j+ \widehat\zeta_t^j e^i \right)\right).$$
As we are in a Gaussian space, the distribution of $\zeta_t$ under $\mathbb{P}^{\mathcal F_t^Z}$ is Gaussian, and  we have\footnote{If $(\xi_1,\xi_2, \xi_3)$ is a Gaussian vector, then $$\mathbb E[\xi_1\xi_2\xi_3] = \mathbb E[\xi_1\xi_2]\mathbb E[\xi_3] + \mathbb E[\xi_1\xi_3]\mathbb E[\xi_2] + \mathbb E[\xi_2\xi_3]\mathbb E[\xi_1] - 2 \mathbb{E}[\xi_1]\mathbb{E}[\xi_2]\mathbb{E}[\xi_3].$$ This result follows easily from the fact that $\mathbb E\left[\left(\xi_1 - \mathbb E[\xi_1]\right)\left(\xi_2 - \mathbb E[\xi_2]\right)\left(\xi_3 - \mathbb E[\xi_3]\right)\right] = 0$ using a symmetry argument between $(\xi_1,\xi_2, \xi_3)$ and $(-\xi_1,-\xi_2, -\xi_3)$. This is a special case of Wick's probability theorem.}
$$\mathbb E\left[\left.\zeta^i_t\zeta^j_t\zeta_t \right|\mathcal F_t^Z \right] = \mathbb E\left[\left.\zeta^i_t\zeta^j_t\right|\mathcal F_t^Z \right]\widehat \zeta_t + \mathbb E\left[\left.\zeta^i_t\zeta_t\right|\mathcal F_t^Z \right]\widehat \zeta^j_t + \mathbb E\left[\left.\zeta^j_t\zeta_t\right|\mathcal F_t^Z \right]\widehat \zeta^i_t - 2 \widehat\zeta^i_t\widehat\zeta^j_t\widehat\zeta_t,$$
hence
\begin{eqnarray*}
\phi^{i,j}_t &=& {\Sigma_Z^{-\frac 12}}\Gamma \begin{pmatrix} 0\\  \mathbb E\left[\left.\zeta^i_t\zeta_t\right|\mathcal F_t^Z \right]\widehat \zeta^j_t + \mathbb E\left[\left.\zeta^j_t\zeta_t\right|\mathcal F_t^Z \right]\widehat \zeta^i_t - 2 \widehat\zeta^i_t\widehat\zeta^j_t\widehat\zeta_t   \end{pmatrix} + \tilde \rho\Sigma_\zeta^{\frac 12}\left(\widehat\zeta_t^i e^j+ \widehat\zeta_t^j e^i \right)\\
&=& {\Sigma_Z^{-\frac 12}}\Gamma \begin{pmatrix} 0\\  \mathbb E\left[\left.\zeta_t\zeta'_t\right|\mathcal F_t^Z \right]-  \widehat\zeta_t\widehat\zeta_t   \end{pmatrix}\left(\widehat\zeta_t^i e^j+ \widehat\zeta_t^j e^i \right) + \tilde \rho\Sigma_\zeta^{\frac 12}\left(\widehat\zeta_t^i e^j+ \widehat\zeta_t^j e^i \right)\\
&=& \psi_t'\left(\widehat\zeta_t^i e^j+ \widehat\zeta_t^j e^i \right)
\end{eqnarray*}
which completes the proof.
\end{proof}

\end{document}